\documentclass[a4,10pt]{article}
\usepackage[left=26mm,top=26mm,right=26mm,bottom=26mm]{geometry}
\usepackage{amssymb,amsthm}
\usepackage{graphicx}
\usepackage{amsmath}
\usepackage{natbib}
\usepackage{siunitx}
\usepackage{booktabs}
\usepackage{lscape}
\usepackage{latexsym,indentfirst,ascmac}
\usepackage{comment}%
\usepackage{color}%

\usepackage{bm}

\newtheorem{lemma*}{Lemma A.}
\newtheorem{definition}{Definition}

\newtheorem{proposition}{Proposition}

\theoremstyle{remark}
\newtheorem{remark}{Remark}
\newtheorem{example}{Example}

\newcommand{\argmax}{\mathop{\rm argmax}\limits}

\newcommand{\tra}[1]{#1^{\top}}

\newcommand{\bX}{\bm{X}}
\newcommand{\bY}{\bm{Y}}

\newcommand{\bI}{\bm{I}}
\newcommand{\bV}{\bm{V}}

\newcommand{\bx}{\bm{x}}
\newcommand{\by}{\bm{y}}
\newcommand{\bz}{\bm{z}}
\newcommand{\bw}{\bm{w}}

\newcommand{\bc}{\bm{c}}

\newcommand{\bmu}{\bm{\mu}}

\newcommand{\bone}{\bm{1}}
\newcommand{\bzero}{\bm{0}}

\newsavebox{\circlebox}
\savebox{\circlebox}{\fontencoding{OMS}\selectfont\Large\char13}
\newlength{\circleboxwdht}

  \newcommand{\Slash}[1]{{\ooalign{\hfil#1\hfil\crcr\raise.167ex\hbox{/}}}}

\newcommand*{\rev}[1]{\textcolor{black}{#1}}

\newcommand{\ou}[3]{%
  \mathrel{%
    \vcenter{\offinterlineskip
      \ialign{##\cr$#1$\cr\noalign{\kern-#3}$#2$\cr}%
    }%
  }%
}
\newcommand*{\omu}[3]{\underset{#3}{\overset{#1}{#2}}}
\newcommand*{\T}{^{\top}}

\newcommand*{\iidsim}{\omu{\text{\tiny{iid}}}{\sim}{}}

\newcommand*{\deq}{\omu{\text{\tiny{d}}}{=}{}}

\newcommand*{\IN}{\mathbb{N}}

\newcommand*{\IR}{\mathbb{R}}

\newcommand*{\Par}{\operatorname{Par}}
\newcommand*{\Dir}{\operatorname{Dir}}

\newcommand*{\U}{\operatorname{U}}

\newcommand*{\N}{\operatorname{N}}

\newcommand{\AC}[1]{\text{AC}_{#1}}
\newcommand*{\Leb}[2]{\operatorname{Leb}_{#1}{(#2)}}

\newcommand*{\rd}{\mathrm{d}}

\renewcommand*{\P}{\mathbb{P}}
\newcommand*{\bQ}{\mathbb{Q}}
\newcommand*{\E}{\mathbb{E}}

\newcommand*{\Var}{\operatorname{Var}}
\newcommand*{\Cov}{\operatorname{Cov}}

\newcommand*{\VaR}[2]{\operatorname{VaR}_{#1}(#2)}

\newcommand*{\ES}[2]{\operatorname{ES}_{#1}(#2)}

\newcommand*{\eps}{\varepsilon}

\newcommand{\bt}{\bm{t}}

\newcommand{\bZ}{\bm{Z}}

\newcommand{\bU}{\bm{U}}

\newcommand{\bbm}{\bm{m}}

\newcommand{\blambda}{\bm{\lambda}}

\newcommand{\bK}{\bm{K}}

\newcommand{\sqc}[2]{#1_1,\dots,#1_{#2}}
\newcommand{\supp}{\operatorname{supp}}

\title{Modality for Scenario Analysis and Maximum Likelihood Allocation}

\author{Takaaki Koike\footnote{Corresnponding author: Department of Statistics and Actuarial Science, University of Waterloo, Waterloo, ON, Canada, E-mail: tkoike@uwaterloo.ca} \ and\
Marius Hofert\footnote{Department of Statistics and Actuarial Science, University of Waterloo, Waterloo, ON, Canada, E-mail: marius.hofert@uwaterloo.ca}}

\begin{document}

\maketitle
\date{}

\begin{abstract}
  \rev{We study the variability of a risk from the statistical viewpoint of multimodality of the conditional loss
  distribution given that the aggregate loss equals an
  exogenously provided capital.}  This conditional distribution serves as a building
  block for calculating risk allocations such as the Euler capital allocation of
  Value-at-Risk. A \rev{superlevel set} of this conditional distribution can be interpreted
  as a set of severe and plausible stress scenarios the given capital is
  supposed to cover.  We show that various distributional properties of this
  conditional distribution\rev{, such as modality, dependence and tail behavior,} are inherited from those of the underlying joint loss
  distribution.  Among these properties, we find that modality of the
  conditional distribution is an important feature \rev{in risk assessment} related to the \rev{variety} of
  risky scenarios likely to occur in a stressed situation.  
  \rev{Under unimodality, we introduce} a novel risk allocation method called maximum
  likelihood allocation (MLA), defined as the mode of the conditional
  distribution given the total capital.  
  \rev{Under multimodality, a single vector of allocations can be less sound.
  To overcome this issue, we investigate the so-called multimodalty adjustment to increase the soundness of risk allocations.}
  Properties of the conditional distribution, MLA \rev{and multimodality adjustment} are
  demonstrated in numerical experiments. In particular, we observe that negative
  dependence among losses typically leads to multimodality, and thus \rev{a higher multimodality adjustment can be required}. 
\end{abstract}

\hspace{1mm}


\section{Introduction}\label{sec:introduction}
\emph{Risk allocation} concerns the quantification of the risk of each unit of a
portfolio. For a $d$-dimensional portfolio of risks or losses represented by an $\IR^d$-valued random vector
$\bX=(\sqc{X}{d})$, $d \in \IN$, the overall loss $S=X_1+\cdots+X_d$ is
covered by a total capital $K \in \IR$, which is typically determined as
$K=\varrho(S)$ for a risk measure $\rho$.  The \emph{Euler principle}, proposed
in \citet{tasche1999risk}, is one of the most well-known rules of risk
allocation.  It is economically justified, for example, in
\citet{tasche1999risk} and \citet{kalkbrener2005axiomatic}, and the derived
allocated capital is also known as the \emph{Aumann-Shapley value}
\citep{aumann2015values} in cooperative game theory; see
\citet{denault2001coherent}.

The Euler principle is applicable when the total capital is determined by a risk measure via $K=\varrho(S)$.
However, as pointed out by \citet{asimit2019efficient}, the total capital in practice may not always coincide with the risk measure itself but includes various adjustments such as stress scenarios and liquidity adjustments.
In such cases, the capital does not possess the original meaning as a risk measure and the formula under the Euler principle is not straightforwardly applicable.
In addition, there are situations when the total capital is given exogenously as a constant; see \citet{laeven2004optimization}.
For the case when the total capital is regarded as a constant, various allocation methods have been proposed in the literature.
One of the main streams found, for example, in \citet{laeven2004optimization} and \citet{dhaene2012optimal}, is to derive an allocation as a minimizer of some loss function over a set of allocations $\mathcal K_d(K)=\{\bx \in \IR^d: x_1+\cdots+x_d=K\}$.
Another method is to find a confidence level for which the corresponding risk
measure coincides with $K$, and then allocate $K$ by regarding it as measured by
a risk measure. For example, if Value-at-Risk (VaR) or Expected Shortfall (ES)
are chosen as risk measures, confidence levels
$p_{\operatorname{VaR}},\ p_{\operatorname{ES}} \in (0,1)$ are first found such
that $K=\VaR{p_{\operatorname{VaR}}}{S}$ or, respectively,
$K=\ES{p_{\operatorname{ES}}}{S}$ hold for a given total capital $K$.  After
performing this procedure, the Euler principle becomes applicable and
the resulting \rev{allocated capital of $K$ to the $j$th risk $X_j$ is $\E[X_j\ | \ \{S=K\}]$ or, respectively, $\E[X_j\ | \ \{S\geq \VaR{p_{\operatorname{ES}}}{S}\}]$; see Section~\ref{subsec:capital:allocation} for details.}

Although these methods provide plausible risk allocations, they sometimes ignore important distributional properties of $\bX$ related to the soundness of risk allocations and to risky scenarios expected to be covered by the allocated capitals.
As we will see in Section~\ref{subsec:motivating:example}, most allocation methods provide the homogeneous allocation $(K/d,\dots,K/d)$ when $\bX$ is exchangeable in the sense that $\smash{\bX\deq (X_{\pi(1)},\dots,X_{\pi(d)})}$ for any permutation $(\pi(1),\dots,\pi(d))$ of $\{1,\dots,d\}$.
This homogeneous allocation can be sound when the conditional distribution of $\bX$ in a stressed situation is unimodal with the mode $(K/d,\dots,K/d)$ since this homogeneous allocation covers the risky scenario most likely to occur in a stressed situation.
On the other hand, the same allocation $(K/d,\dots,K/d)$ is derived when the conditional distribution in a stressed situation is multimodal and $(K/d,\dots,K/d)$ is supposed to cover multiple risky scenarios \emph{on average}.
\rev{In this multimodal case, the homogeneous allocation can be less sound than in the former unimodal case since multiple risky scenarios are hidden in a single vector of allocations.
Therefore, the soundness of risk allocations depends on the distributional properties of the conditional distribution of $\bX$ in a stressed situation. 
In the multimodal case, imposing a multimodality loading to the capital $K$ can be required to take the variability of scenarios into account.}

In this paper, \rev{we study the variability of a risk from the statistical viewpoint of multimodality of} the conditional distribution of $\bX$ given
$\{S=K\}$.  Since $\bX\ | \ \{S=K\}$ takes values in $\mathcal K_d(K)$, this
random vector serves as a building block for deriving risk allocations.  For
example, \rev{
the Euler allocation of $\VaR{p}{S}$, $p \in (0,1)$, arises as the expectation of $\bX\ | \ \{S=K\}$ with 
$K=\VaR{p}{S}$.
In addition,} we show in
Section~\ref{subsec:motivating:example}
that a \rev{superlevel set} of $\bX\ | \ \{S=K\}$ can be regarded as a set of
severe and plausible stress scenarios the given capital $K$ is supposed to
cover. Based on the motivation provided there, we investigate distributional
properties of $\bX\ | \ \{S=K\}$ in
Section~\ref{sec:properties:X:given:constant:sum}.  We show that unimodality, dependence and 
tail behavior of $\bX\ | \ \{S=K\}$ are typically inherited from
those of the underlying unconditional loss $\bX$, respectively.  Moreover, we
demonstrate in Section~\ref{subsec:simulation:study} that negative dependence among $\bX$ typically leads
to multimodality of $\bX\ | \ \{S=K\}$.  These observations can be useful to
detect the hidden risk of multimodality in risk allocation.
The properties of $\bX\ | \ \{S=K\}$ studied in this paper are of potential importance in simulation and statistical inference of $\bX\ | \ \{S=K\}$ using Markov chain Monte Carlo (MCMC) methods for efficiently simulating the distribution of interest; see Remark~\ref{remark:simulation:conditional:distribution:MCMC} and Appendix~\ref{app:simulation:with:MCMC}.

We also propose a novel risk allocation method termed \emph{maximum likelihood
  allocation (MLA)}, which is defined as the mode of
$\bX\ | \ \{S=K\}$ \rev{assuming that it is unimodal}.  Besides the mean (which leads to the Euler allocation of VaR), the mode is also an important summary statistics of
$\bX\ | \ \{S=K\}$. It can be interpreted as the risky scenario most likely to occur in the
stressed situation $\{S=K\}$.  By searching for the global mode of
$\bX\ | \ \{S=K\}$, possibly multiple local modes can be detected. As explained
in
Section~\ref{subsec:motivating:example},
this procedure of detecting multimodality is beneficial for evaluating the
soundness of risk allocations, for discovering hidden multiple scenarios likely to
occur in the stressed situation $\{S=K\}$ and for constructing more flexible
risk allocations by weighting important scenarios. 
In Section~\ref{subsec:maximum:likelihood:allocation} we present a definition and properties of MLA expected to hold for a risk allocation.
\rev{In Section~\ref{subsec:multimodality:adjustment} we introduce the so-called multimodality adjustment to increase the soundness of risk allocations when $\bX\ | \ \{S=K\}$ is multimodal.
The multimodality adjustment can be regarded as a measure of variability of the underlying risk; see \citet{furman2006tail} and \citet{furman2017gini} for studies of variability measures of tail risk.}
MLA \rev{and multimodality-adjusted allocated capitals are} estimated and compared with Euler allocations in numerical experiments in Section~\ref{sec:numerical:experiments}.
Concluding remarks are given in Section~\ref{sec:discussion:conclusion} and all proofs can be found in Appendix~\ref{app:proofs}.

\section{Preliminaries}\label{sec:preliminaries}

\subsection{A brief introduction to capital allocation}\label{subsec:capital:allocation}

On a standard atomless probability space $(\Omega,\mathcal A,\P)$, let $\bX=(X_1,\dots,X_d)$, $d\geq 2$ be a $d$-dimensional random vector with joint distribution function $F_{\bX}$ with margins $F_{X_1},\dots,F_{X_d}$ and a copula $C$.
Furthermore, let $S=X_1+\cdots+X_d$ and denote $F_S$ by its distribution function.
If $F_S$ and $F_{\bX}$ have densities, we denote them by $f_S$ and $f_{\bX}$, respectively, with marginal densities $f_{X_1},\dots,f_{X_d}$ of $f_{\bX}$ and copula density $c$.
The variable $X_j$ is interpreted as loss of the $j$th asset, business line, economic entity and so on, of the portfolio $\bX$ in a fixed period of time.
Similarly, $S$ is regarded as the aggregate risk of the portfolio $\bX$.
Positive values of $\sqc{X}{d}$ and $S$ are understood as losses and negative values are interpreted as profits.

The amount of total capital required to cover the risk of the portfolio $\bX$ is often determined as $\varrho(S)$ where $\varrho$ is a \emph{risk measure}, that is, a map from a set of random variables to a real number.
Examples of risk measures include \emph{Value-at-Risk (VaR)} at confidence level $p \in (0,1)$ defined by
\begin{align*}
  \VaR{p}{X}=\inf\{x \in \IR: F_X(x) \ge p\},
\end{align*}
for a random variable $X$ on $(\Omega,\mathcal A,\P)$ and its distribution function $F_X$, and \emph{Expected Shortfall (ES)} at confidence level $p \in (0,1)$, also known as \emph{Conditional VaR, Tail VaR and Average VaR}, defined by
\begin{align*}
  \ES{p}{X}=\frac{1}{1-p}\int_{p}^{1}\VaR{q}{X}\,\rd q,
\end{align*}
provided that $\E[|X|]<\infty$.

Once the total capital is determined as $K\in\IR$, it is decomposed into $d$ real numbers $\AC{1},\dots,\AC{d}$ such that the \emph{full allocation property}
\begin{align}\label{eq:full:allocation:property}
\AC{1}+\cdots+\AC{d}=K
\end{align}
holds.
The set of all possible allocations is denoted by
\begin{align*}
\mathcal K_d(K): = \{\bx \in \IR^d: x_1+\cdots+x_d = K\}.
\end{align*}
If $K=\varrho(S)$ for a \rev{positive homogeneous} risk measure $\varrho$, the so-called \emph{Euler principle} determines the $j$th allocated capital by
\begin{align*}
\AC{j}^{\operatorname{Euler}}=\left. \frac{
  \partial \varrho(\tra{\blambda} \bX)
}{
  \partial \lambda_j
}\right|_{\blambda=\bone_d},\quad\text{where}\quad
\bone_d=(1,\dots,1)\in \IR^d,
\end{align*}
\rev{provided that the partial derivative exists; see \citet{tasche2001conditional} for more details on the differentiability argument.
The Euler principle} leads to the \emph{VaR contributions} and \emph{ES contributions} given by
\begin{align}\label{eq:VaR:contributions}
\left. \frac{
  \partial \varrho(\tra{\blambda} \bX)
}{
  \partial \lambda_j
}\right|_{\blambda=\bone_d} = \E[X_{j}\ |\ \{ S=\VaR{p}{S}\}]\quad
\text{ when }\quad \varrho = \operatorname{VaR}_{p},
\end{align}
and
\begin{align}\label{eq:ES:contributions}
\left. \frac{
  \partial \varrho(\tra{\blambda} \bX)
}{
  \partial \lambda_j
}\right|_{\blambda=\bone_d}= \E[X_{j}\ |\ \{S\geq \VaR{p}{S}\}],\quad
\text{ when }\quad \varrho = \operatorname{ES}_{p},
\end{align}
respectively.

We consider the case when the capital is an exogenously given constant $K \in \IR$.
Our proposed risk allocation introduced in Section~\ref{sec:maximum:likelihood:allocation} is based on the conditional distribution
\begin{align}\label{eq:conditional:distribution:constant:sum}
F_{\bX|\{S=K\}}(\bx)=\P(\bX \leq \bx\ | \ \{S=K\}),\quad  \bx \in \IR^d.
\end{align}
The conditional distribution~\eqref{eq:conditional:distribution:constant:sum} is degenerate and its first $d'=d-1$ components $\bX'\ |\ \{S=K\}=(X_1,\dots,X_{d'})\ |\ \{S=K\}$ determine the last one via $X_d\ |\ \{S=K\} = K - (X_1+\cdots+ X_{d'})\ |\ \{S=K\}$.
Therefore, it suffices to consider the $d'$-dimensional marginal distribution $F_{\bX'|\{S=K\}}$.
Note that throughout this paper, the $'$-notation is used to denote quantities
  related to this non-degenerate distribution in $d-1$ dimensions and should not be confused with
  matrix transposition for which we will use the $\T$-symbol.
Assuming that $\bX$ and $(\bX',S)$ admit densities, $\bX'\ |\ \{S=K\}$ also has a density and is given by
\begin{equation}\label{eq:conditional:density:X:constant:sum}
f_{\bX'|\{S=K\}}(\bx')=\frac{f_{(\bX',S)}(\bx',K)}{f_{S}(K)}
=\frac{f_{\bX}(\bx',K-\tra{\bone_{d'}}\bx')}{f_{S}(K)},\quad  \bx' \in \IR^{d'},
\end{equation}
where the last equality follows from an affine transformation $(\bX',S)\mapsto \bX$ with unit Jacobian.

\subsection{A motivating example}\label{subsec:motivating:example}

The distribution of $\bX\ |\ \{S=K\}$ is a primary subject in this paper.
In this section, we provide a motivating example for investigating this distribution from the viewpoint of risk alloation.

\begin{figure}[t]
  \centering
  \vspace{-35mm}
  \includegraphics[width=15 cm]{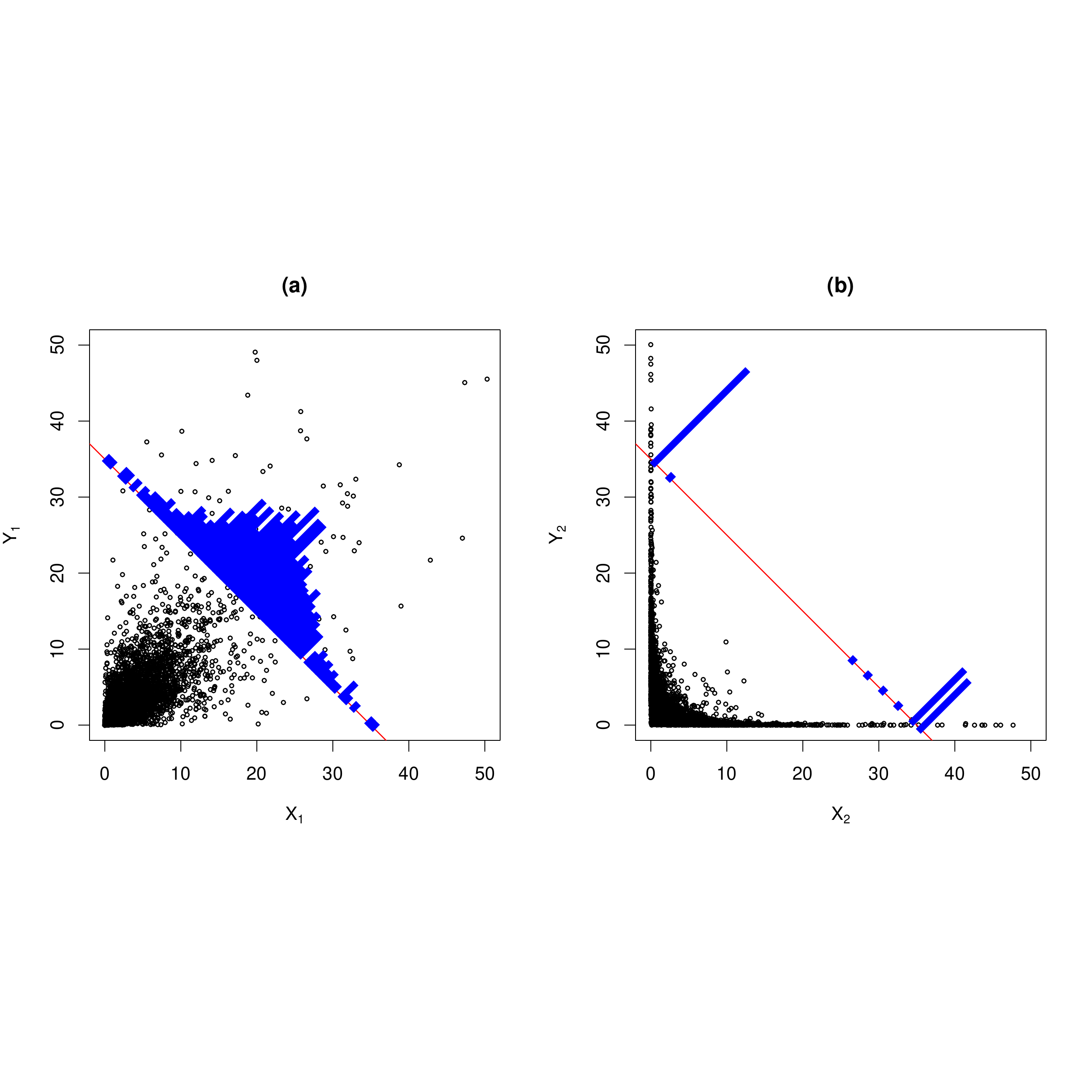}\vspace{-35mm}
  \caption{Scatter plots (black dots) of (a) $(X_1,Y_1)$ and (b) $(X_2,Y_2)$ such that all of $X_1,Y_1,X_2$ and $Y_2$ identically follow Pareto distributions with shape parameter 3 and scale parameter 5, and $(X_1,Y_1)$ and $(X_2,Y_2)$ have Student $t$ copulas $C_{\nu,\rho_1}^t$ and $C_{\nu,\rho_2}^t$, respectively, where $\nu=5$ is the degrees of freedom, and $\rho_1=0.8$ and $\rho_2=-0.8$ are the correlation parameters.
  The red line indicates $x + y =K$ for $K=35$.
  Histograms (blue) of the conditional distributions of (a) $(X_1,Y_1)$ and (b) $(X_2,Y_2)$ on the (approximate) set of allocations $\{(x,y) \in \IR^2:K-\delta < x+y < K + \delta\}$, $\delta = 0.5$, are drawn on $\mathcal K_d(K)=\{(x,y) \in \IR^2:x+y= K \}$.
  }
  \label{fig:plots:comparison:pos:neg}
\end{figure}

\subsubsection{\rev{Soundness of risk allocations}}\label{subsubsection:soundness:risk:allocations}

Consider two bivariate risks (a) $(X_1,Y_1)$ and (b) $(X_2,Y_2)$ such that all of $X_1,Y_1,X_2$ and $Y_2$ identically follow Pareto distributions with shape parameter 3 and scale parameter 5, and $(X_1,Y_1)$ and $(X_2,Y_2)$ have Student $t$ copulas $C_{\nu,\rho_1}^t$ and $C_{\nu,\rho_2}^t$, respectively, where $\nu=5$ is the degrees of freedom parameter and $\rho_1=0.8$ and $\rho_2=-0.8$ are the correlation parameters. Suppose that the exogenously given total capital equals $K=35$.
By exchangeability of the risk models (a) and (b), most allocation rules provide the homogeneous allocation $(K/2,K/2)=(17.5,17.5)$ in both cases (a) and (b).
For instance, if $K$ is regarded as VaR or ES at some confidence levels and is allocated according to the Euler principle, then both VaR and ES contributions lead to homogeneous allocations.
As we see in Figure~\ref{fig:plots:comparison:pos:neg}, however, the conditional distributions of $(X_1,Y_1)$ and of $(X_2,Y_2)$ on the set of allocations $\mathcal K_d(K)$ differ substantially.
Positive dependence among $X_1$ and $Y_1$ prevents the two random variables from moving in opposite directions under the constraint $X_1 + Y_1 = K$, which results in unimodality of the conditional distribution on $\mathcal K_d(K)$.
On the other hand, negative dependence among $X_2$ and $Y_2$ allows them to move in opposite directions, which leads to bimodality of the conditional distribution.
From the viewpoint of risk management, the homogeneous allocation $(K/2,K/2)$ seems to be a more sound capital allocation in Case (a) because it covers the most likely risky scenario.
In Case (b), the two risky scenarios around the corners $(K,0)$ and $(0,K)$ occur equally likely and the allocation $(K/2,K/2)$ can be understood as an average of these scenarios.
However, the likelihood around $(K/2,K/2)$ is quite small and a single vector of the equal allocation $(K/2,K/2)$ obscures the two distinct risky scenarios.
\rev{Moreover, either of $X_2$ or $Y_2$ is likely to suffer a large amount of loss if $(K/2,K/2)$ is reserved as capital.}
Consequently, the soundness of the allocated capital depends on the modality of the conditional loss distribution, and multiple risky scenarios can be hidden in a single vector of risk allocations

\rev{
\begin{remark}[Multimodality and variability]
Models (a) and (b) may not be directly comparable since the total capital for Model (b) is typically smaller than that for Model (a).
However, the variability of a risk under Model (b) is larger than under Model (a) in the sense that a wider variety of scenarios are likely to occur under Model (b) than under Model (a).
Therefore, if the total capital $K$ does not incorporate the variability of a risk, an adjustment of $K$ may be required to take the multimodality of scenarios into account.
\end{remark}
}

\subsubsection{\rev{Stress testing of risk allocations}}\label{subsubsec:stress:testing:risk:allocations}

Inspecting modes of $\bX\ |\ \{S=K\}$ can also be regarded as a stress test of
risk allocations. \citet{breuer2018find} requires stress scenarios to be severe
and plausible.
\rev{
We define the \emph{scenario set with a level of plausibility $t>0$} by $L_t(\bX)=\{\bx\in \IR^d: f_{\bX}(\bx)\geq t\}$ where $f_{\bX}$ is the density function of $\bX$ assuming that it exists. 
Among the scenario set $L_t(\bX)$, the set $L_t(\bX)\cap \mathcal K_d(K)$ can be regarded as the set of the most severe scenarios the given total capital $K$ can cover.
Using the convention $f_{\bX|\{S=K\}}(\bx)=f_{\bX}(\bx)\bone_{\{\bone_d\T\bx = K\}}/f_S(K)$, $\bx \in \IR^d$, the set $L_t(\bX)\cap \mathcal K_d(K)$ leads to the superlevel set of $\bX\ | \ \{S=K\}$ with level $t/f_S(K)$ since 
\begin{align*}
L_t(\bX)\cap \mathcal K_d(K)&=\{\bx \in \IR^d: f_{\bX}(\bx)\bone_{\{\bone_d\T\bx = K\}}\geq t\}\\
&= \{\bx \in \IR^d: f_{\bX|\{S=K\}}(\bx)\geq t/f_S(K)\}
=L_{t/f_S(K)}(\bX\ | \ \{S=K\}).
\end{align*}
Throughout the paper, the superlevel set of $\bX\ | \ \{S=K\}$ is treated as a set of \emph{stress scenarios}.
In particular, the modes of $\bX\ |\ \{S=K\}$ are the most severe and plausible scenarios that $K$ can cover since they attain the highest level of plausibility among the stress scenarios.}
Unimodality of $\bX\ |\ \{S=K\}$ \rev{(see Definition~\ref{def:concepts:unimodality} for its formal definition)} implies that
there exists one representative stress scenario the total capital $K$ can cover, and thus
the mode can be a sound risk allocation covering the risky scenario most likely to occur.
On the other hand, multimodality of $\bX\ |\ \{S=K\}$ means that there are
multiple distinct stress scenarios that are severe and plausible, and thus it may not be
sufficient to only focus on a single scenario without identifying the other ones.


\begin{remark}[Simulation of $\bX\ |\ \{S=K\}$ with MCMC methods]
\label{remark:simulation:conditional:distribution:MCMC}
Another motivation for investigating distributional properties of
$\bX\ |\ \{S=K\}$ is an efficient simulation of this conditional
distribution.  This is a challenging task in general since there is no general and
tractable sampling method known for $\bX\ |\ \{S=K\}$.  
Although samples from $\bX\ |\ \{S=K\}$ can be generated by first simulating $\bX$ and then extracting those satisfying the constraint $\{S=K\}$,
the probability $\P(S=K)$ is zero, and thus such samples
virtually never exist when $S$ admits a density.  A potential remedy for this
problem is to modify the conditioning set $\{S=K\}$ to
$\{K-\delta <S<K+\delta\}$ for a small $\delta>0$ so that
$P(K-\delta <S<K+\delta)>0$.  However, this modification distorts the
distribution of $\bX\ |\ \{S=K\}$ and the resulting estimates of risk
allocations suffer from inevitable biases.  To overcome this issue, \citet{koike2019estimation} and
\citet{koike2020markov} proposed MCMC methods for exact simulation from
$\bX\ |\ \{S=K\}$.  Although MCMC methods improve sample efficiency and the
resulting estimates are unbiased, their performance highly depends on
distributional properties of $\bX\ |\ \{S=K\}$, in particular on its modality
and tail behavior; see Appendix~\ref{app:simulation:with:MCMC} for more details.  From this viewpoint,
investigating properties of $\bX\ |\ \{S=K\}$ is of potential importance to construct
efficient MCMC methods for simulating $\bX\ |\ \{S=K\}$.
\end{remark}

\section{Properties of the conditional distribution given a constant sum}
\label{sec:properties:X:given:constant:sum}

\rev{In Section~\ref{subsec:motivating:example}, we showed that the conditional distribution of $\bX$ given a constant sum $\{S=K\}$ plays an important role in risk allocations.
With this motivation in mind,} we study the support, modality, dependence and tail behavior of $\bX\ | \ \{S=K\}$ in this section. 
As introduced in Section~\ref{subsec:capital:allocation}, we consider the $d'$-dimensional random vector $\bX'\ |\ \{S=K\}$ for $d'=d-1$ to avoid the degeneracy of $\bX\ | \ \{S=K\}$.

\subsection{Support of $\bX\ |\ \{S=K\}$}\label{subsec:support:conditional:distribution:given:sum}

We start with the support of $\bX'|\{S=K\}$.
Assuming that $\bX'|\{S=K\}$ admits a density $f_{\bX'|\{S=K\}}$, 
Equation~\eqref{eq:conditional:density:X:constant:sum} implies that
\begin{align*}
\supp(\bX'\ |\ \{S=K\})=\{\bx' \in \IR^{d'}: f_{\bX'|\{S=K\}}(\bx')>0\}=
\{\bx' \in \IR^{d'}: f_{\bX}(\bx',K-\bone_{d'}\T\bx')>0\}.
\end{align*}
Therefore, if $\sqc{X}{d}$ are supported on $\IR^d$, then $\supp(\bX'\ |\ \{S=K\})=\IR^{d'}$.
Another typical case is when $\supp(\bX)=(l_1,\infty)\times \cdots \times (l_d,\infty)$ for some $\sqc{l}{d}> - \infty$, which implies that $\sqc{X}{d}$ are bounded from below, that is, $X_j \geq l_j$ $\P$-almost surely (a.s.), for $j=1,\dots,d$.
In this case, the support of $\bX'|\{S=K\}$ is given by
\begin{align}\label{eq:support:conditional:distribution:constant:sum}
\supp(\bX'\ |\ \{S=K\})=\biggl\{\bx' \in \IR^{d'}: x_1 > l_1,\dots, x_{d'}> l_{d'},\ \sum_{j=1}^{d'}x'_j < K -l_d  \biggl\}.
\end{align}

If $l_1=\cdots=l_{d}=0$, that is, when $\bX$ models the nonnegative part of
losses, then the closure of~\eqref{eq:support:conditional:distribution:constant:sum}
is known as the \emph{$K$-simplex}. Since the
set in \eqref{eq:support:conditional:distribution:constant:sum} is bounded,
simulation of $\bX'\ |\ \{S=K\}$ can be more straightforward than in the former case when $\supp(\bX'\ |\ \{S=K\})=\IR^{d'}$.
For instance, an independent Metropolis-Hastings (MH) algorithm can be applied by first generating a sample $\by'$ uniformly on the set in \eqref{eq:support:conditional:distribution:constant:sum}
and then replacing the current state $\bx'$ with the new state $\by'$ with probability $\alpha(\bx',\by')=f_{\bX'|\{S=K\}}(\by')/f_{\bX'|\{S=K\}}(\bx')=f_{\bX}(\by',K-\bone_{d'}\T\by')/f_{\bX}(\bx',K-\bone_{d'}\T\bx')$.


\subsection{$\bX\ |\ \{S=K\}$ in the elliptical case}\label{subsec:conditional:distribution:constant:sum:elliptical:case}

\emph{Elliptical distributions} are important exceptions for which the distribution of $\bX'\ |\ \{S=K\}$ can be derived explicitly.
See Chapter 6 of \citet{mcneil2015quantitative} for applications of elliptical distributions to risk management.
Throughout this work, the set of all $d\times d$ positive definite matrices is denoted as $\mathcal M_{+}^{d\times d}$.
The characteristic function of a random vector $\bX$ is given by $\phi_{\bX}(\bt)=\E[\exp(i \bt\T\bX)]$, $\bt \in \IR^d$.
If a function $\psi(t): [0,\infty) \rightarrow \IR$ is such that $\psi(\bt\T\bt)$ is a $d$-dimensional characteristic function, then $\psi$ is called a \emph{characteristic generator}; see \citet{fang2018symmetric} for details.
Let $\Psi_d$ denote the class of all characteristic generators.
A $d$-dimensional random vector $\bX$ is said to have an \emph{elliptical distribution}, denoted by $\bX \sim \mathcal E_d(\bmu,\Sigma,\psi)$, if its characteristic function can be expressed as
\begin{align*}
\phi_{\bX}(\bt)=\exp(i\bt \T\bmu)\ \psi\left(\frac{1}{2}\bt\T\Sigma\bt\right)
\end{align*}
for a \emph{location vector} $\bmu \in \IR^d$, \emph{dispersion matrix} $\Sigma \in \mathcal M_{+}^{d\times d}$ and a \emph{characteristic generator} $\psi \in \Psi_d$.
When an elliptical distribution $\bX \sim \mathcal E_d(\bmu,\Sigma,\psi)$ admits a density function, it is of the form
\begin{align*}
f_{\bX}(\bx)=\frac{c_d}{\sqrt{|\Sigma|}}g\left(\frac{1}{2}(\bx-\bmu)\T\Sigma^{-1}(\bx-\bmu); d\right),\quad \bx \in \IR^d,
\end{align*}
for some normalizing constant $c_d>0$ and a \emph{density generator} $g(\cdot; d)$ satisfying
\begin{align*}
\int_0^{\infty}t^{d/2 -1}g(t;d)\,\rd t< \infty;
\end{align*}
see \citet{fang2018symmetric}.
We omit the second argument and write $g(\cdot)=g(\cdot\ ; d)$ when it can be ignored.

In the following proposition we derive the distribution of $\bX'\ |\ \{S=K\}$ provided that $\bX \sim \mathcal E_d(\bmu,\Sigma,\psi)$.

\begin{proposition}[Ellipticality of $\bX'\ |\ \{S=K\}$]\label{prop:ellipticality:conditional:distribution:constant:sum}
Suppose $\bX \sim \mathcal E_d(\bmu,\Sigma,\psi)$.
Then $\bX'\ |\ \{S=K\}$ follows an elliptical distribution $\mathcal E_{d'}(\bmu_K,\Sigma_K,\psi_K)$ for some characteristic generator $\psi_K \in \Psi_{d'}$,
\begin{align}\label{eq:location:dispersion:parameter:conditional:dist:constant:sum}
\bmu_K= \bmu' + \frac{K - \mu_S}{\sigma_S^2}(\Sigma\bone_{d})'\quad \text{and}\quad
\Sigma_K = {\Sigma'}- \frac{1}{\sigma_S^2}(\Sigma\bone_{d})'\tra{(\Sigma\bone_{d})'},
\end{align}
where $\bmu'$ and $(\Sigma\bone_d)'$ are the first $d'$-components of $\bmu$ and $(\Sigma\bone_d)$, respectively, $\Sigma'$ is the principal submatrix of $\Sigma$ deleting the $d$th row and column,
 $\mu_S=\bone_d \T \bmu$ and $\sigma_S^2 = \bone_d\T \Sigma \bone_d$.
Furthermore, if $\bX$ admits a density with density generator $g$, then $\bX'\ |\ \{S=K\}$ admits a density with density generator
\begin{align}\label{eq:density:generator:conditional:distribution:constant:sum}
g_K(t)=g(t  + \Delta_K)\quad \text{where} \quad \Delta_K =  \frac{1}{2}\left(\frac{K-\mu_S}{\sigma_S}\right)^2.
\end{align}
\end{proposition}

Note that the characteristic generator $\psi_K$ of $\bX'\ |\ \{S=K\}$ is in
general different from that of $\bX$; see the proof in
Appendix~\ref{app:proofs}.  By
Proposition~\ref{prop:ellipticality:conditional:distribution:constant:sum},
ellipticality is preserved under conditioning $\{S=K\}$ and thus a change of the
shape of the distribution as observed in
Figure~\ref{fig:plots:comparison:pos:neg} (b) does not occur when $\bX$ is
elliptical.  The capital $K$ is typically much larger than the mean of the total
loss $\mu_S$ in
practice. 
Therefore, by~\eqref{eq:density:generator:conditional:distribution:constant:sum},
the density generator $g_K$ is typically the tail part of the generator
$g$.  Moreover, the location vector $\bmu_K$ typically increases in proportion to the sum
of covariances $(\Sigma \bone_d)'$.
As a consequence, more (less) capital is assigned to losses which are positively (negatively) correlated with the other losses.
On the other hand, the dispersion matrix $\Sigma_K$ decreases in proportion to the term $(\Sigma \bone_d)'\tra{(\Sigma \bone_d)'}$ and the reduction depends on the variance of the sum.

\begin{example}[Student $t$ distribution]\label{example:student:t:conditional:distribution:sum:constraint}
A $d$-dimensional \emph{Student $t$ distribution} $t_{\nu}(\bmu,\Sigma)$ is an elliptical distribution $\mathcal E_d(\bmu,\Sigma,\psi)$ with density generator
\begin{align}\label{eq:density:generator:student:t}
g(t; d)=\left(1+\frac{t}{\nu}\right)^{-\frac{d+\nu}{2}},\quad t \geq 0,
\end{align}
where $\nu \geq 1$ is the \emph{degrees of freedom} parameter.
It is known, for example, from \citet{roth2012multivariate} and \citet{ding2016conditional} that the conditional distribution of the Student $t$ distribution is again Student $t$.
We can check this closedness property with Proposition~\ref{prop:ellipticality:conditional:distribution:constant:sum}.
By~\eqref{eq:density:generator:conditional:distribution:constant:sum}, the random variable $\bX'\ |\ \{S=K\}$ follows an elliptical distribution $\mathcal E_{d'}(\bmu_K,\Sigma_K,g_K)$ with density generator (up to a constant) given by
\begin{align*}
g_K(t)=\left(1+\frac{t}{\nu + \Delta_K}\right)^{-\frac{d+\nu}{2}},
\end{align*}
for which the corresponding distribution is known as the \emph{Pearson type $V\hspace{-1mm}I\hspace{-1mm}I$ distribution}; see \citet{schmidt2002tail}.
In fact, this distribution reduces to a $d'$-dimensional Student $t$ distribution since
\begin{align*}
g_K(t)=\left(1+\frac{t}{\nu + \Delta_K}\right)^{-\frac{d+\nu}{2}}\propto \left(1+\frac{\nu+1}{\nu + \Delta_K}\frac{t}{\nu+1}\right)^{-\frac{d'+\nu+1}{2}},
\end{align*}
and the multiplier $(\nu+1)/(\nu + \Delta_K)$ can be absorbed by redefining the dispersion matrix as $\tilde \Sigma_K = (\nu + \Delta_K) \Sigma_K / (\nu+1)$ for $(\nu + \Delta_K) / (\nu+1) >0$.
Consequently, $\bX'\ |\ \{S=K\}$ follows $t_{\nu+1}(\bmu_K,\tilde \Sigma_K)$.
Since the degrees of freedom of $\bX'\ |\ \{S=K\}$ increases by $1$, $\bX'\ |\ \{S=K\}$ has slightly lighter tails than $\bX$.
\end{example}

\subsection{Unimodality of $\bX\ |\ \{S=K\}$}\label{subsec:unimodality:conditional:distribution:constant:sum}

Next we study the modality of $\bX'\ |\ \{S=K\}$.
Among various definitions of unimodality considered in the literature, we adopt those defined based on the \rev{superlevel set}
\begin{align*}
L_t(f)=\{\bx\in\IR^d: f(\bx)\geq t\},\quad t \in (0,\  \max\{f(\bx):\bx \in\IR^d\}],
\end{align*}
where $f$ is a density on $\IR^d$ which is assumed to be bounded for simplicity so that $\max\{f(\bx):\bx \in\IR\}$ exists.
By definition, $L_t(f)$ is a decreasing set, that is, $L_{t'}(f)\subseteq L_{t}(f)$ for $0<t \leq t'$.
We also write $L_t(\bX)$ for $L_t(f)$ if $\bX$ has density $f$.
A set $A\subseteq \IR^d$ is called \emph{star-shaped} about $\bx_0 \in A$ if, for any $\by \in A$, the line segment from $\bx_0$ to $\by$ is in $A$.

\begin{definition}[Concepts of unimodality]\label{def:concepts:unimodality}
For a bounded density function $f$ on $\IR^d$, we call $M(f) = L_{t^\ast}(f)$ the \emph{mode set} of $f$ for $t^\ast=\max\{f(\bx):\bx \in\IR^d\}$.
If $L_{t^{\ast}}(f)=\{\bbm\}$ then we call $\bbm \in \IR^d$ the \emph{mode} of $f$.
Furthermore, $f$ is said to be \emph{weakly unimodal} if $L_t(f)$ is connected, \emph{star unimodal} about the center $\bx_0 \in \IR^d$ if $L_t(f)$ is star-shaped about $\bx_0$ and \emph{convex unimodal} if $L_t(f)$ is convex, for all $0<t\leq t^\ast$.
\rev{Finally, $f$ is said to be \emph{multimodal} if $L_t(f)$ is not connected for some $0<t \leq t^{\ast}$.}
\end{definition}

From Definition~\ref{def:concepts:unimodality}, convex unimodality implies star unimodality and star unimodality implies weak unimodality.
Other notions of unimodality, such as block unimodality, linear unimodality, monotone unimodality, $\alpha$-unimodality, orthounimodality and Khinchin's unimodality are not introduced in this paper due to their intractability for our purpose; see \citet{dharmadhikari1988unimodality} for a comprehensive discussion on unimodality.
As mentioned in
Section~\ref{subsec:motivating:example},
$L_t(\bX)$ can be understood as a plausible scenario set with $t>0$ being the
level of plausibility.  In addition, $L_t(\bX\ | \ \{S=K\})$ can be regarded as
a set of severe and plausible stress scenarios the total capital $K$ is supposed
to cover.  From these interpretations, we believe that the notion of unimodality should
describe tractability of these scenario sets, such as connectivity and convexity.
The \rev{superlevel set} $L_t(f)$ can also be important when $f$ is simulated with MCMC methods
since the ratio of levels of $f$ is a primary quantity of interest for such
methods.  MCMC methods are required to be specifically designed when $L_t(f)$ is
not connected since in this case a Markov chain needs to traverse distinct
regions to simulate samples from the entire space.

Note that uniqueness of the maximum of a density $f$, that is, the mode set of $f$ being a singleton $L_{t^{\ast}}(f)=\{\bbm\}$ for $\bbm \in \IR^d$, is an important but different concept of unimodality from those in Definition~\ref{def:concepts:unimodality}.
In fact, uniqueness of the maximum may not be an appropriate concept of unimodality when the relationship between $\bX$ and $\bX'\ | \ \{S=K\}$ is of interest. 
This is because the uniqueness of the maximum of $f_{\bX'|\{S=K\}}$ is equivalent to that of $f_{\bX}$ on the restricted domain $\mathcal K_d(K)$ via~\eqref{eq:conditional:density:X:constant:sum}, and thus the uniqueness of the maximum of $f_{\bX}$ on the entire support $\IR^d$ does not provide any information on the shape of $f_{\bX}$ on $\mathcal K_d(K)$ unless the mode of $f_{\bX}$ on $\IR^d$ is in $\mathcal K_d(K)$.

The following proposition reveals some relationships between unimodality of $\bX$ and that of $\bX'\ |\ \{S=K\}$.

\begin{proposition}[Unimodality of $\bX'\ |\ \{S=K\}$]\label{prop:unimodality:conditional:distribution:constant:sum}
\hspace{3mm}
\begin{enumerate}
\item Suppose $\bX\sim \mathcal E_d(\bmu,\Sigma,\psi)$ admits a density with density generator $g$.
If $g$ is decreasing on $\IR_{+}$, then $f_{\bX'|\{S=K\}}$ is convex unimodal.
Furthermore, if the equation $g(t)=\Delta_K$ of $t \in \IR_{+}$ has a unique solution $t_K^\ast$, then $f_{\bX'|\{S=K\}}$ has the mode $\bbm=\bmu_K$.
\item If $\bX$ is convex unimodal, then $\bX'\ |\ \{S=K\}$ is convex unimodal.
\end{enumerate}
\end{proposition}

Unlike convex unimodality, neither weak unimodality nor star unimodality of $\bX$ implies any of the unimodality concepts of $\bX'\ |\ \{S=K\}$ introduced in Definition~\ref{def:concepts:unimodality}.
To provide a counterexample, we introduce the following class of distributions.

\begin{definition}[Homothetic density]\label{def:homothetic:distribution}
A $d$-dimensional random vector $\bX$ is said to have a \emph{homothetic density}, denoted by $\bX \sim \mathcal H(\bmu,D,r)$, with a \emph{location parameter} $\bmu\in\IR^d$, \emph{shape set} $D\subseteq \IR^d$ and a \emph{scaling function} $r: \IR_{+} \rightarrow \IR_{+}$ if $\bX-\bmu$ admits a density $f_D$ satisfying
\begin{align*}
L_t(f_D)=r(t)D=\{s\bx: 0\leq s \leq r(t),\ \bx \in D\}
\end{align*}
for some continuous and decreasing function $r$ and a bounded and star-shaped set $D \in \IR^d$ around $\bzero$ such that
\begin{align}\label{eq:condition:r:D:homothetic:distribution}
 \int_0^\infty\Leb{d}{r(t)D}\,\rd t=1,
 \end{align}
 where $\operatorname{Leb}_d$ denotes the Lebesgue measure on $\IR^d$.
\end{definition}

Note that Condition~\eqref{eq:condition:r:D:homothetic:distribution} is required to ensure that $\int_{\IR^d} f_D(\bx)\,\rd \bx =1$.
To see this, we have
\begin{align*}
\int_{\IR^d}f_D(\bx)\,\rd \bx &=
\int_{\IR^d} \int_0^{f_D(\bx)}\ \rd t \,\rd \bx
=
\int_{\IR^d} \int_0^\infty \bone_{\{\bx \in L_t(f_D)\} }\,\rd t\,\rd \bx\\
& = \int_{0}^\infty\Leb{d}{L_t(f_D)}\,\rd t
=  \int_0^\infty\Leb{d}{r(t)D}\,\rd t=1.
\end{align*}

Homothetic distributions arise from $l_p$-spherical distributions \citep{osiewalski1993robust} where the \rev{superlevel set}s are determined as balls in the $l_p$-norm, and also arise from a further generalized class of distributions called the $v$-spherical distributions \citep{fernandez1995modeling}.
Examples of homothetic distributions include skew-normal distributions and rotund-exponential distributions; see \citet{balkema2010asymptotic}.
It is straightforward to check that $\bX\sim \mathcal H(\bzero_d,D,r)$ is star unimodal about $\bx_0\in\IR^d$ if $D$ is star-shaped about $\bx_0$, and convex unimodal if $D$ is convex.

Suppose $\bX \sim \mathcal H(\bzero_d,D,r)$ for a convex set $D$.
Then $\bX$ is convex unimodal and so is $\bX'\ |\ \{S=K\}$ by Proposition~\ref{prop:unimodality:conditional:distribution:constant:sum} Part 2.
For this homothetic distribution, the \rev{superlevel set} of $\bX'\ |\ \{S=K\}$ embedded in $\IR^d$ has the following representation
\begin{align*}
\{\bx \in \IR^d& : \bx' \in L_t(\bX'\ |\ S=K),  \ x_d = K - \bone_{d'}\T\bx'\}\\
& =
\{\bx \in \IR^{d}: f_{\bX'|\{S=K\}}(\bx')\geq t, \ x_d = K - \bone_{d'}\T\bx'\}\\
&= \{\bx \in \IR^{d}: f(\bx)\geq t  f_S(K) \} \cap \mathcal K_d(K)
= L_{t  f_S(K)}(f) \cap \mathcal K_d(K)\\
&= r(t  f_S(K) )D \cap \mathcal K_d(K)
=\{s \bx : \bx \in D,\ 0 \leq s \leq r(t  f_S(K) )\}\cap \mathcal K_d(K)\\
&= \left\{\frac{K}{\bone_d\T\bx}\bx: \bx \in D,\ 0\leq \frac{K}{\bone_d\T\bx}\leq r(t  f_S(K) )\right\}\\
&= \left\{\frac{K}{\bone_d\T\bx}\bx: \bx \in \bigcup_{k \geq K/r(t  f_S(K) )}D \cap \mathcal K_d(k)\right\},
\end{align*}
that is, the \rev{superlevel set} $L_t(\bX'\ |\ \{S=K\})$ embedded in $\IR^d$ is a collection of the projected points of $\bx \in D$ intersected with the upper half space $\{\bx \in \IR^d: \bone_d\T\bx \geq K/r(t  f_S(K) )\}$ onto $\mathcal K_d(K)$.

The following example shows that neither weak unimodality nor star unimodality of $\bX$ implies any of the unimodality concepts introduced in Definition~\ref{def:concepts:unimodality} for $\bX'\ |\ \{S=K\}$.

\begin{example}
Consider $\bX \in \mathcal H(\bzero_2,D,r)$ where $D=([-2,2]\times [-1,1])\cup([-1,1]\times [-2,2])$ and $r(t)=\frac{1}{2\sqrt{3}}\exp(-t/2)$.
$D$ is star-shaped (and thus connected) around $(0,0)$ and $r$ is a decreasing function.
Furthermore, the pair of $(D,r)$ satisfies Condition~\eqref{eq:condition:r:D:homothetic:distribution} since
\begin{align*}
\int_0^\infty\Leb{2}{r(t)D}\,\rd t = \Leb{2}{D}\int_0^\infty r^2(t)\,\rd t
= 12  \int_0^\infty \frac{1}{12}\exp(-t)\,\rd t =1.
\end{align*}
Suppose that the total capital is given by $K=1/3$.
For $t=-2\log(\sqrt{3}/3)\approx1.098$, we have $r(t)=1/6$ and thus $L_t(f_D)=D/6 = ([-1/3,1/3]\times [-1/6,1/6])\cup([-1/6,1/6]\times [-1/3,1/3])$.
Therefore, $L_t(\bX'\ |\ \{S=K\})=[0,1/6]\cup[1/3,1/2]$, which is neither star-shaped nor even connected.
\end{example}

Next we study marginal properties of unimodality.
In general, even if $\bX$ is convex unimodal, it does not imply any unimodality for its marginal distributions; see \citet[Example~A.3.]{balkema2010asymptotic} for a counterexample.
The following example shows that marginal unimodality also does not imply joint unimodality.

\begin{example}[Marginal unimodality does not imply joint unimodality]
Consider the following bivariate density
\begin{align*}
f(u,v)=\frac{9}{4} \bone_{\left\{(u,v)\in \bigcup_{i=1}^3[(i-1)/3,i/3]^2  \right\}} + \frac{9}{4}  \bone_{\left\{(u,v)\in [1/3,2/3]^2\right\}}, \quad (u,v)\in[0,1],
\end{align*}
which has the convex unimodal marginal densities
\begin{align*}
f_1(u)=f_2(u)=\frac{3}{4} \bone_{\{u \in [0,1]\}} + \frac{3}{4} \bone_{\{u \in [1/3,2/3]\}},\quad u \in [0,1].
\end{align*}
However, $L_{9/4}(f)=[0,1/3]^2\cup[1/3,2/3]^2\cup[2/3,1]^2$ is neither convex nor star-shaped.
\end{example}

Joint unimodality implies marginal unimodality for certain classes of distributions.
As is shown in \citet{balkema2010asymptotic}, $l_p$-spherical distributions form a subclass of homothetic densities for which unimodality is preserved under marginalization.
This property also holds for the class of $s$-concave densities, which is also closed under the operation $\bX \mapsto \bX'\ |\ \{S=K\}$; see Appendix~\ref{app:s:concave:densities} for details.


\subsection{Dependence of $\bX\ |\ \{S=K\}$ and its modality}\label{subsec:dependence:conditional:distribution:constant:sum}

The dependence structure of $\bX'\ |\ \{S=K\}$ is typically described in terms of the dependence among $X_j$ and $S$ for $j=1,\dots,d'$.
For instance, when $\bX \sim \mathcal E_d(\bmu,\Sigma,\psi)$,
Proposition~\ref{prop:ellipticality:conditional:distribution:constant:sum} yields
\begin{align*}
\Cov[X_i,X_j\ |\ \{S=K\}] & =(\Sigma_K)_{i,j} = \Cov[X_i,X_j]-\frac{1}{\sigma_S^2}(\Sigma\bone_d)_i(\Sigma\bone_d)_j\\
&= \Cov[X_i,X_j]-\frac{1}{\sigma_S^2}\Cov[X_i,S]\Cov[X_j,S]
=\sigma_i \sigma_j (\rho_{X_i,X_j}-\rho_{X_i,S}\ \rho_{X_j,S}),
\end{align*}
where $\sigma_j^2 = \Var(X_j)$ and $\rho_{X_i,X_j}$ is the correlation coefficient of $(X_i,X_j)$.
\rev{Beyond the elliptical case, various dependence concepts, in particular the total positivity and its related order of $\bX'\ |\ \{S=K\}$, are investigated in Appendix~\ref{app:dependence:stochastic:order}.}

\rev{In this section, we investigate the modality of $\bX\ |\ \{S=K\}$ under extremal dependence of $\bX$.
The following example shows that $\bX\ |\ \{S=K\}$ is degenerate and thus unimodal when $\bX$ is comonotone.
}

\rev{
\begin{example}[$\bX'\ |\ \{S=K\}$ under comonotonicity]\label{ex:conditional:distribution:constant:sum:comonotonicity}
Suppose $\bX$ is a comonotone random vector with continuous and strictly increasing margins $\sqc{F}{d}$, that is, $\smash{\bX\deq (F_1^{-1}(U),\dots,F_d^{-1}(U))}$ for some $U \sim \U(0,1)$.
Under continuity and strict increasingness of $\sqc{F}{d}$, their quantile functions $F_j^{-1}$, $j=1,\dots,d$, are continuous and strictly increasing.
Therefore, the equation $\sum_{j=1}^d F_j^{-1}(u)=K$ of $u \in [0,1]$ has a unique solution $u^\ast$.
Moreover, it holds that
\begin{align*}
\bX\ |\ \{S=K\}=(F_1^{-1}(u^{\ast}),\dots,F_d^{-1}(u^{\ast}))\quad \P\text{-a.s.},
\end{align*}
since
\begin{align*}
\P\biggl( \ \bigcup_{j=1}^d\{X_j\neq F_j^{-1}(u^{\ast})\}  \biggl| \biggl \{S=K\biggr\}\biggr)\biggr.
=\P\biggl(\ \bigcup_{j=1}^d\left\{F_{j}^{-1}(U)\neq F_j^{-1}(u^{\ast})\right\} \biggl| \biggl\{\sum_{j=1}^dF_{j}^{-1}(U)=K\biggr\}\biggr)\biggr.=0.
\end{align*}
This example can be understood as an extreme case where positive dependence
(comonotonicity) implies unimodality of $\bX\ |\ \{S=K\}$, which takes on one point $(F_1^{-1}(u^{\ast}),\dots,F_d^{-1}(u^{\ast}))$ with
probability $1$.
\end{example}
}

When $\bX$ has negative dependence, a wider variety of
distributions, possibly multimodal ones, arise as $\bX\ |\ \{S=K\}$ compared with the positive dependent case.
\rev{We demonstrate this phenomenon in the following example.}


\begin{example}[$\bX\ |\ \{S=K\}$ under extreme negative dependence]\label{ex:multimod}
Let $K>0$ and $X\sim F$ for a continuous distribution function $F$ supported on $[0,\infty)$ such that $X\ | \ \{X\leq K\}$ is radially symmetric about $K/2$ in the sense that $\smash{(X-K/2)\ | \ \{X\leq K\}\deq (K/2 - X)\ | \ \{X\leq K\}}$.
For $U\sim \U(0,1)$ define $(X_1,X_2)$ by
\begin{align*}
X_1 &=F^{-1}(U) \bone_{\{U\leq F(K)\}} + F^{-1}(U)\bone_{\{U>F(K)\}}= F^{-1}(U),\\
X_2 &=(K-F^{-1}(U)) \bone_{\{U\leq F(K)\}} + F^{-1}(U)\bone_{\{U>F(K)\}}.
\end{align*}
Then $\P(X_1\leq x)=\P(F^{-1}(U)\leq x)=\P(U\leq F(x))=F(x)$ for all $x \geq 0$.
Since the conditional radial symmetry of $F$ implies that $\P(K-F^{-1}(U)\leq x,\ U\leq F(K))=\P(F^{-1}(U)\leq x,\  U\leq F(K))$, it holds that
\begin{align*}
\P(X_2\leq x)&=\P(X_2\leq x,\ U\leq F(K))+\P(X_2\leq x,\ U> F(K))\\
&=\P(K-F^{-1}(U)\leq x,\ U\leq F(K))+\P(F^{-1}(U)\leq x,\ U> F(K))\\
&=\P(F^{-1}(U)\leq x,\ U\leq F(K))+\P(F^{-1}(U)\leq x,\ U> F(K))\\
&=\P(F^{-1}(U)\leq x)=F(x),\quad x \geq 0.
\end{align*}
Therefore, $X_1\sim F$ and $X_2\sim F$.
The body $\{X_1\leq K\}$ of $X_1$ and the tail $\{X_2>K\}$ of $X_2$ are mutually exclusive in the sense that $\P(X_1\leq K,\ X_2>K)=0$.
Similarly $\P(X_1> K,\ X_2\leq K)=0$.
In the tail, $X_1$ and $X_2$ are comonotone in the sense that $(X_1,X_2)=(F^{-1}(U),F^{-1}(U))$ on $\{U>F(K)\}$.
In the body, $X_1$ and $X_2$ are countermonotone in the sense that $(X_1,X_2)=(F^{-1}(U),K-F^{-1}(U))$ on $\{U\leq F(K)\}$.
Since $X_1+X_2 = F^{-1}(U) + K-F^{-1}(U)=K$ on $\{U\leq F(K)\}$ and $X_1+X_2 =2 F^{-1}(U)>2K > K$ on $\{U>F(K)\}$, we have that
\begin{align*}
  \{X_1+X_2=K\}=\{X_1+X_2=K,\ U\leq F(K)\} \cup \{X_1+X_2=K,\ U >  F(K)\} = \{U\leq F(K)\},
\end{align*}
and thus that
\begin{align*}
(X_1,X_2) \ | \ \{X_1+X_2=K\} = (X_1,X_2) \ | \  \{U\leq F(K)\} = (F^{-1}(U),K-F^{-1}(U)) \ | \  \{U\leq F(K)\}.
\end{align*}
Consequently, $(X_1,X_2)\ | \ \{S=K\}$ has the homogeneous marginal distribution $F_{X|\{X\leq K\}}$ and a countermonotone copula $W$.
Therefore, multimodality of $\bX\ | \ \{S=K\}$ appears when, for example, $X\sim F$ has a bimodal distribution on the body $\{X\leq K\}$; see Remark~\ref{rem:extension:complete:mixability} for a more concrete example.
\end{example}

\begin{remark}[Extension to $d\geq 3$ with complete mixability]\label{rem:extension:complete:mixability}
  Example~\ref{ex:multimod} for constructing $(X_1,X_2)$ based on
  countermonotonicity can be extended to the multivariate case $d\geq 3$.  
  Let $K>0$ and
  $X\sim F$ for a continuous distribution function $F$ supported on $[0,\infty)$
  such that the conditional distribution $F_{X|\{X \leq K\}}$ is
  \emph{$d$-completely mixable} with center $K$ for $d\geq 3$, that is, there
  exists a $d$-dimensional random vector $\bY=(\sqc{Y}{d})$ called the
  \emph{$d$-complete mix} such that $Y_j \sim F_{X|\{X \leq K\}}$, $j=1,\dots,d$, and
  $Y_1+\cdots + Y_d=K$ a.s.
  Such a random vector exists, for example,
  when $F_{X|\{X \leq K\}}$ admits a decreasing density with $\E[Y_1]=K/d$; see
  \citet[Corollary~2.9.]{wang2011complete}.  Define $\bX=(\sqc{X}{d})$ by
  $X_j=Y_j \bone_{\{U\leq F(K)\}} + Z_j \bone_{\{U>F(K)\}}$ for
  $\bY=(\sqc{Y}{d})$ being the $d$-complete mix of $F_{X|\{X \leq K\}}$, $U \sim \U(0,1)$, $Z_j \sim F_{X|\{X > K\}}$, $j=1,\dots,d$ and $\bY$, $U$ and $\sqc{Z}{d}$ are independent of each
  other.  Then one can check that $X_j \sim F$.  Moreover,
  $\{X_1+\cdots + X_d = K \}=\{U\leq F(K)\}$ since
\begin{align*}
S=X_1+\cdots + X_d = K  \bone_{\{U\leq F(K)\}} + (Z_1+\cdots+Z_d) \bone_{\{U>F(K)\}},
\end{align*}
and $Z_1+\cdots+Z_d > dK>K$ a.s.
Consequently, $\bX\ | \ \{S=K\}=\bX\ | \ \{U\leq F(K)\}=\bY$ a.s.\ and thus $\bX \ | \ \{S=K\}$ is the $d$-complete mix of $X\ | \ \{X\leq K\}$.
To construct a multimodal $\bX\ | \ \{S=K\}$ one can choose $\bY$ as an equally weighted mixture of three Dirichlet distributions $\Dir(\alpha,\alpha,\beta)$, $\Dir(\alpha,\beta,\alpha)$ and $\Dir(\beta,\alpha,\alpha)$ for $0<\alpha < \beta$.
This mixture is a $3$-complete mix since it has homogeneous marginal distributions and a constant sum.
Moreover, $\bY$ has non-connected superlevel sets when, for example, $\alpha=2$ and $\beta=10$, and thus $\bX'\ | \ \{S=K\}$ is multimodal.
\end{remark}


\subsection{Tail behavior of $\bX\ |\ \{S=K\}$}\label{subsec:tail:dependence:conditional:distribution:constant:sum}

In this section we study the tail behavior of $\bX'\ |\ \{S=K\}$ through its density.
Since boundedness of $\bX$ from below leads to a bounded support of $\bX'\ |\ \{S=K\}$ as shown in Section~\ref{subsec:support:conditional:distribution:given:sum}, we focus on the case when $\bX$ is supported on $\IR^d$.
In this case, the support of $\bX'\ |\ \{S=K\}$ is $\IR^{d'}$ and thus there are $2^{d'}$ orthants to be considered.
Hereafter we consider tail behavior only in the first orthant $\{\bx' \in \IR^{d'}:x_1,\dots,x_{d'}> 0\}$ since tails on the other orthants can be discussed similarly.
We study the following limiting behaviors of the ratio of densities.

\begin{definition}[Multivariate regular and rapid variations of a density]\label{def:multivariate:regular:variation:density}
Let $\bX$ be a $d$-dimensional random vector $\bX$ with a density $f_{\bX}$.
\begin{enumerate}
\item  $\bX$ is called \emph{multivariate regularly varying} with limit function $\lambda: \IR^{2d} \rightarrow \IR_{+}$ (at $\infty$ and on the first orthant), denoted by MRV($\lambda$), if
\begin{align}\label{eq:multivariate:regularly:varying:density}
\lim_{t\rightarrow \infty} \frac{f_{\bX}(t\by)}{f_{\bX}(t\bx)}=\lambda(\bx,\by)> 0\quad  \text{for any}\quad \bx,\by \in \IR_{+}^{d},
\end{align}
provided the limit function $\lambda$ exists.
\item $\bX$ is called \emph{multivariate rapidly varying} (at $\infty$ and on the first orthant), denoted by MRV$(\infty)$ if,
\begin{align*}
\lim_{t\rightarrow \infty} \frac{f_{\bX}(st\bx)}{f_{\bX}(t\bx)}=
\begin{cases}
0, & s>1,\\
\infty, & 0< s < 1,\\
\end{cases}\quad  \text{for any}\quad s >0\quad\text{and}\quad  \bx \in \IR_{+}^{d}.
\end{align*}
\end{enumerate}
\end{definition}

Note that we adopt the definitions of regular and rapid variations of densities for their potential application to MCMC methods where the ratio of target densities $f_{\bX'|\{S=K\}}(\by')/f_{\bX'|\{S=K\}}(\bx')$ at any two points $\bx',\ \by' \in \IR^{d'}$ is of interest; see Appendix~\ref{app:simulation:with:MCMC}.
Taking $\bx=\bone_d$ in~\eqref{eq:multivariate:regularly:varying:density} leads to the standard definition of regular variation introduced, for example, in \citet{resnick2007heavy}.
Regular variation is typically described in terms of probability measures or their survival functions, and these concepts of variations are connected to regular variation of densities through \citet[Theorem~6.4.]{resnick2007heavy}.

The following proposition states that one can find a limit function for $\bX'\ |\ \{S=K\}$ based on that of $\bX$ through the auxiliary random vector $\tilde \bX = (\bX',K-X_d)$.

\begin{proposition}[Multivariate regular and rapid variations of $\bX'\ |\ \{S=K\}$]\label{prop:tail:density:conditional:distribution:constant:sum}
\hspace{2mm}
\begin{enumerate}
  \item Assume that $\tilde \bX= (\bX',K-X_d)$ is MRV($\tilde \lambda$).
Then $\bX'\ |\ \{S=K\}$ is MRV($\lambda'$) with limit function
\begin{align*}
\lambda'(\bx',\by')= \tilde \lambda((\bx',\bone_{d'}\T\bx'),(\by',\bone_{d'}\T\by')),\quad \bx',\by'\in \IR^{d'}_{+}.
\end{align*}
\item If $\tilde \bX$ is MRV($\infty$), then $\bX'\ |\ \{S=K\}$ is MRV($\infty$).
\end{enumerate}
\end{proposition}

The sufficient conditions in
Proposition~\ref{prop:tail:density:conditional:distribution:constant:sum} are
straightforward to check 
since
$\tilde \bX$ does not depend on the sum $S$, and the joint distribution of
$\tilde \bX$ can be specified through its marginal distributions and copula.
The margins of $\tilde \bX$ are $\tilde F_j=F_j$, $j=1,\dots,d'$, and
$\tilde F_d(x_d)=\bar F_d(K-x_d)$, and the copula $\tilde C$ of $\tilde \bX$ is
the distribution function of $(U_1,\dots,U_{d'},1-U_d)$ where $\bU\sim C$ is the
copula of $\bX$.  
This enables one to find a limit function for $\tilde \bX$; see, for example, \citet{li2013toward}
and \citet{joe2019tail}.

The following proposition shows that the limit function is determined by the density generator $g$ in the elliptical case.

\begin{proposition}[Multivariate regular and rapid variations for elliptical distribution]\label{prop:mrv:conditional:distribution:elliptical}
Assume $\bX \sim {\mathcal E}_d(\bmu,\Sigma,\psi)$ admits a density with density generator $g$ continuous on $\IR_{+}$.
\begin{enumerate}
\item If $g$ is regularly varying in the sense that
\begin{align*}
\lim_{t\rightarrow \infty}\frac{g(tu)}{g(ts)}=\lambda_g(s,u),\quad s, u >0,
\end{align*}
then $\bX'\ |\ \{S=K\}$ is MRV($\lambda_K$) with
\begin{align*}
\lambda_K(\bx',\by')=\lambda_g({\bx'}\T\Sigma_K^{-1}\bx',\ {\by'}\T\Sigma_K^{-1}\by'),\quad \bx',\ \by' \in \IR^{d'}.
\end{align*}
\item If $g$ is rapidly varying in the sense that
\begin{align*}
\lim_{t\rightarrow \infty}\frac{g(st)}{g(t)}=\begin{cases}
0, & s>1,\\
\infty, & 0<s < 1,\\
\end{cases}
\end{align*}
then $\bX'\ |\ \{S=K\}$ is MRV($\infty$).
\end{enumerate}
\end{proposition}


\begin{example}[Normal and Student $t$ distributions]
\label{example:mrv:densities}
The multivariate Normal distribution has a rapidly varying density generator $g(t)=\exp(-t)$, and thus its corresponding conditional distribution $\bX'\ |\ \{S=K\}$ is also rapidly varying by Proposition~\ref{prop:mrv:conditional:distribution:elliptical} Part 2.
Next, suppose $\bX$ follows a $d$-dimensional Student $t$ distribution with degrees of freedom $\nu \geq 1$.
Its density generator~\eqref{eq:density:generator:student:t} is regularly varying with limit function
\begin{align*}
\lim_{t \rightarrow \infty}\frac{g(tu)}{g(ts)}=\left(\frac{u}{s}\right)^{-\frac{\nu+d}{2}}, \quad u,s>0.
\end{align*}
Consequently, by Proposition~\ref{prop:mrv:conditional:distribution:elliptical} Part 1, $\bX'\ |\ \{S=K\}$ is regularly varying with limit function
\begin{equation*}
\lim_{t \rightarrow \infty} \frac{f_{\bX'|\{S=K\}}(t\by')}{f_{\bX'|\{S=K\}}(t\bx')}=\left(
\frac{
||\Sigma_K^{-\frac{1}{2}}\by'||
}{
||\Sigma_K^{-\frac{1}{2}}\bx'||
}
\right)^{-(\nu+d)},\quad \bx',\by' \in \IR^{d'}_{+},
\end{equation*}
where $||\cdot||$ is an Euclidean norm on $\IR^{d'}$.
\end{example}

\section{Maximum likelihood allocation \rev{and multimodality adjustment}}\label{sec:maximum:likelihood:allocation}

\rev{In this section we investigate how the modality of $\bX\ | \ \{S=K\}$ can be incorporated in risk management.
Under unimodality, the mode of $\bX\ | \ \{S=K\}$ is regarded as the most likely stress scenario covered by the given total capital $K$.
This mode is defined to be a maximum likelihood allocation (MLA) in Section~\ref{subsubsec:definition:assumptions}, and its properties are studied in Section~\ref{subsubsec:properties:mla}.
Under multimodality of $\bX\ | \ \{S=K\}$, a single vector of allocations may obscure multiple risky scenarios as seen in Section~\ref{subsec:motivating:example}.
To overcome this issue, 
we introduce the so-called multimodality adjustment in Section~\ref{subsec:multimodality:adjustment} 
to utilize the knowledge of multimodality of $\bX\ | \ \{S=K\}$ and to increase the soundness of risk allocations.
}


\subsection{Maximum likelihood allocation}\label{subsec:maximum:likelihood:allocation}

\subsubsection{\rev{Definition and assumptions on MLA}}\label{subsubsec:definition:assumptions}

We denote by $\mathcal U_d(K)$ the set of all
$d$-dimensional random vectors $\bX$ such that $\bX$ and $(\bX',S)$ admit
density functions, and that the function $\bx \mapsto f_{\bX}(\bx)\bone_{\{\bx \in \mathcal K_d(K)\}}$ has a unique maximum.  For $\bX\in \mathcal U_d(K)$, $\bX'\ |\ \{S=K\}$ admits a
density through~\eqref{eq:conditional:density:X:constant:sum}, and
$\bx' \mapsto f_{\bX'|\{S=K\}}(\bx')$ has a unique maximum attained by the mode of $\bX'\ |\ \{S=K\}$.
By
Proposition~\ref{prop:unimodality:conditional:distribution:constant:sum},
elliptical random vectors with continuous and decreasing density generators form
a subclass of $\mathcal U_d(K)$.  Although some exchangeable random vectors
possessing negative dependence, such as Model (b) in
Section~\ref{subsec:motivating:example},
may not be included in $\mathcal U_d(K)$, we believe that most loss models used
in risk management practice are contained in $\mathcal U_d(K)$.  As explained
in Section~\ref{subsec:unimodality:conditional:distribution:constant:sum},
uniqueness of the mode of $\bX'\ |\ \{S=K\}$ and its unimodality are different
concepts, and thus the class $\mathcal U_d(K)$ contains multimodal random
vectors in the sense that the density $f_{\bX'|\{S=K\}}$ has multiple local
maximizers (we call them the \emph{local modes} of $\bX'\ | \ \{S=K\}$\rev{; see Definition~\ref{def:local:modes} for their formal definitions}).
\rev{Nevertheless, in this section 
we focus only on the unique global maximizer of
$f_{\bX'|\{S=K\}}$ (not on local ones) since unimodal distributions are the primary object to apply the MLA to.
The multimodal case will then be revisited in Section~\ref{subsec:multimodality:adjustment}.
As we will demonstrate in
Section~\ref{sec:numerical:experiments}, multimodality can be
detected by searching for the modes of $f_{\bX'|\{S=K\}}$.
}


In the following we define the unique mode of $\bX'\ | \ \{S=K\}$ as a risk allocation of $K$.

\begin{definition}[Maximum likelihood allocation]\label{def:maximum:likelihood:allocation}
For $K>0$ and $\bX \in \mathcal U_d(K)$, the \emph{maximum likelihood allocation (MLA)} on a set $\mathcal K \subseteq \mathcal  K_d(K)$ is defined by
\begin{align*}
\bK_{\operatorname{M}}[\bX;\mathcal K]=\argmax\{f_{\bX}(\bx): \bx \in \mathcal K\},
\end{align*}
provided the function $\bx \mapsto f_{\bX}(\bx)\bone_{\{\bx \in \mathcal K\}}$ has a unique maximum.
When $\mathcal K = \mathcal K_d(K)$,
we call it the maximum likelihood allocation.
\end{definition}

By~\eqref{eq:conditional:density:X:constant:sum}, MLA of $K$ on $\mathcal K$ can be equivalently formulated as
\begin{align*}
\bK_{\operatorname{M}}[\bX;\mathcal K]=\argmax\{f_{\bX'|\{S=K\}}(\bx'): (\bx',K-\bone_{d'}\T\bx') \in \mathcal K\}.
\end{align*}

By definition, MLA on $\mathcal K \subseteq \mathcal  K_d(K)$ is an allocation of $K$ in the sense that it satisfies the full allocation property $\bone_d\T\bK_{\operatorname{M}}[\bX;\mathcal K]=K$.
We mainly study the case when $\mathcal K = \mathcal K_d(K)$.
However, as we will see in Section~\ref{subsubsec:properties:mla} and Appendix~\ref{app:subsec:application:core:allocation}, the set $\mathcal K$ can be taken so that $\bK_{\operatorname{M}}[\bX;\mathcal K]$ satisfies some desirable properties for a risk allocation principle.

\subsubsection{Properties of MLA}\label{subsubsec:properties:mla}

We now investigate properties of MLA as a risk allocation principle.
For desirable properties of risk allocation in the case when the capital $K$ is exogenously given as a constant, see \citet{maume2016capital}.
By construction, $\bK_{\operatorname{M}}[\bX;\mathcal K]$ always satisfies the full allocation property \eqref{eq:full:allocation:property}.
The following proposition summarizes other desirable properties of MLA.

\begin{proposition}[Properties of MLA]\label{prop:properties:mla}
Suppose $K>0$ and $\bX \in \mathcal U_d(K)$.
\begin{enumerate}
\item \emph{Translation invariance}: $\bK_{\operatorname{M}}[\bX+\bc;\ \mathcal K_d(K+\bone_{d}\T\bc)]=\bK_{\operatorname{M}}[\bX;\mathcal K_d(K)]+\bc$ for $\bc \in \IR^d$.
\item \emph{Positive homogeneity}: $\bK_{\operatorname{M}}[c\bX;\mathcal K_d(cK)]=c\bK_{\operatorname{M}}[\bX;\mathcal K_d(K)]$ for $c > 0$.
  \item \emph{Symmetry}: For $(i,j)\in \{1,\dots,d\}$, $i\neq j$, let $\tilde \bX$ be a $d$-dimensional random vector such that $\tilde X_j=X_i$, $\tilde X_i=X_j$ and $\tilde X_k=X_k$, $k \in \{1,\dots,d\}\backslash\{i,j\}$.
  If $\smash{\bX\deq \tilde \bX}$, then $\bK_{\operatorname{M}}[\bX;\mathcal K_d(K)]_i=\bK_{\operatorname{M}}[\bX;\mathcal K_d(K)]_j$, where $\bK_{\operatorname{M}}[\bX;\mathcal K_d(K)]_l$ is the $l$th component of $\bK_{\operatorname{M}}[\bX;\mathcal K_d(K)]$ for $l=1,\dots,d$.
  \item  \emph{Continuity}:
  Suppose $\bX_n, \ \bX \in \mathcal U_d(K)$ have densities $f_n$ and $f$ for $n=1,2,\dots,$ respectively.
  If $f_n$ is uniformly continuous and bounded for $n=1,2,\dots,$ and $\bX_n \rightarrow \bX$ weakly, then $\lim_{n\rightarrow \infty}\bK_{\operatorname{M}}[\bX_n;\mathcal K_d(K)]$ $=\bK_{\operatorname{M}}[\bX;\mathcal K_d(K)]$.
\end{enumerate}
\end{proposition}

Translation invariance states that a sure loss $\bc \in \IR^d$ requires the same amount of risk allocation and the rest of the total capital is allocated to the random loss $\bX$.
Positive homogeneity means that, for a proportion $c>0$, 100$c$\% of the loss $\bX$ requires 100$c$\% of the total capital $K$ and the resulting MLA of $c\bX$ is 100$c$\% of the allocation of $K$ to $\bX$.
Symmetry implies that, if exchanging two units does not change the distribution of the joint loss, then equal amounts of capitals are allocated to them.
Finally, continuity ensures that if MLA is calculated based on an estimated model $f_n$ of $f$, then this estimate of MLA is close to the true MLA.
\rev{Note that the assumption that $\bX_n$, $n=1,2,\dots$, and $\bX$ belong to $\mathcal U_d(K)$ is esssential so that the MLAs of $\bX_n$, $n=1,2,\dots$, and $\bX$ are well-defined.}

Next we discuss properties that need to be considered separately.

\begin{enumerate}
\item \emph{RORAC compatibility and core compatibility}:\vspace{2mm}\\
\emph{RORAC compatibility} and \emph{core compatibility} are important properties of risk allocations since either of them characterizes Euler allocation; see \citet{tasche1999risk} and \citet{denault2001coherent}.
However, the definitions of these properties are not meaningful when $K$ is exogenously given as a constant.
Moreover, similar constraints as in core compatibility can be additionally imposed on $\mathcal K_d(K)$ so that the resulting MLA is core-compatible; see Appendix~\ref{app:subsec:application:core:allocation} for details.

\item \emph{Riskless asset}:\vspace{2mm}\\
The \emph{riskless asset} condition requires the sure loss $X_j = c_j$ a.s.\ for $c_j \in \IR$ to be covered by the amount of allocated capital $c_j$.
This property needs to be considered separately since in this case $\bX$ does not admit a density.
Suppose that $X_j=c_j \in \IR$ a.s.\ for $j \in I \subseteq \{1,\dots,d\}$ and that $\bX_{-I}=(X_j,j\in \{1,\dots,d\}\backslash I)$ admits a density $f_{\bX_{-I}}$.
Since
\begin{align}\label{eq:riskless:asset:equality}
(\bX_{I},\bX_{-I})\ |\ \{S=K\}&\deq(\bc,\bX_{-I})\ |\ \{\bone_{|-I|}\T\bX_{-I}=K-\bone_{|I|}\T\bc\}
\deq (\bc,\bX_{-I}\ |\ \{\bone_{|-I|}\T\bX_{-I}=K-\bone_{|I|}\T\bc\}),
\end{align}
any realization $\bx$ of $\bX\ |\ \{S=K\}$ satisfies $\bx_{I}=\bc$ and the likelihood of $\bx$ is quantified through the value of the density $f_{\bX_{-I}|\{\bone_{|-I|}\T\bX_{-I}=K-\bone_{|I|}\T\bc\}}(\bx_{-I})$.
According to this discussion, a natural extension of the definition of MLA to such a random vector $\bX$ is
\begin{align}\label{eq:definition:mla:extension:constant}
\bK_{\operatorname{M}}[\bX;\mathcal K_d(K)]_{I}=\bc\quad\text{and}\quad
\bK_{\operatorname{M}}[\bX;\mathcal K_d(K)]_{-I}= \bK_{\operatorname{M}}[\bX_{-I};\mathcal K_{|-I|}(K-\bone_{|I|}\T\bc)],
\end{align}
which is compatible with the riskless asset property.

\item \emph{Allocation under comonotonicity}:\vspace{2mm}\\
Suppose $\bX$ is a comonotone random vector with continuous and strictly increasing margins $\sqc{F}{d}$.
As seen in Example~\ref{ex:conditional:distribution:constant:sum:comonotonicity}, it holds that $\bX\ | \ \{S=K\}=(F_1^{-1}(u^{\ast}),\dots,F_d^{-1}(u^{\ast}))$ a.s., where $u^\ast \in [0,1]$ is the unique solution of $\sum_{j=1}^d F_j^{-1}(u)=K$.
According to the extended definition of MLA~\eqref{eq:definition:mla:extension:constant}, we have that
\begin{align*}
\bK_{\operatorname{M}}(\bX;\mathcal K_d(K))=(F_1^{-1}(u^{\ast}),\dots,F_d^{-1}(u^{\ast})).
\end{align*}

\end{enumerate}

\subsubsection{\rev{Discussion on MLA}}\label{subsubsec:discussion:mla}
We now discuss whether MLA is an appropriate risk allocation principle, and also compare MLA with Euler allocation.
Here we define Euler allocation by $\E[\bX\ | \ \{S=K\}]$, which are VaR contributions~\eqref{eq:VaR:contributions} with $K=\VaR{p}{S}$ for some confidence level $p\in(0,1)$.
As shown in Proposition~\ref{prop:properties:mla}, MLA possesses properties naturally required as an allocation such as translation invariance, positive homogeneity and riskless asset.
Euler allocation also satisfies these properties since $\E[\bX+\bc\ | \ \{\bone_d\T(\bX+\bc)=K+\bone_d\T\bc\}]=\E[\bX\ | \ \{\bone_d\T\bX=K\}]+\bc$ for $\bc \in \IR^d$ (translation invariance), $\E[c\bX\ | \ \{\bone_d\T(c\bX)=cK\}]=c\E[\bX\ | \ \{\bone_d\T\bX=K\}]$ for $c>0$ (positive homogeneity) and the riskless asset property holds by taking expectation on the both sides of the first equality in~\eqref{eq:riskless:asset:equality}.
\rev{See Appendix~\ref{app:properties:not:hold:MLA} for properties that neither MLA nor Euler allocation satisfy.}
Note that by Proposition~\ref{prop:ellipticality:conditional:distribution:constant:sum} and Proposition~\ref{prop:unimodality:conditional:distribution:constant:sum} Part 1, Euler and maximum likelihood allocations coincide when $\bX$ is elliptically distributed.
\rev{Therefore, the economic justifications of Euler allocation, such as RORAC compatibility and core-compatibility, also hold for MLA when $\bX$ is elliptical.}
\rev{Moreover, through the process of estimating a MLA, one can detect multimodality of $\bX'\ | \ \{S=K\}$ and discover hidden risky scenarios based on which one can evaluate the soundness of risk allocations.}
On the other hand, the main disadvantage of MLA compared with Euler allocation is that estimating modes becomes more difficult than estimating a mean as the dimension of the portfolio becomes larger.
\rev{Furthermore, MLA is not well-defined for distributions whose $\argmax\{f_{\bX}(\bx): \bx \in \mathcal K_d(K)\}$ is not a single point.
Finally, MLA may ignore the behavior of $\bX\ | \ \{S=K\}$ other than its mode.}
Considering these aspects, we believe that MLA \rev{itself may not be an appropriate risk allocation principle, but its estimation procedure is beneficial
for discovering hidden multiple scenarios likely to occur in the stressed situation $\{S=K\}$, for assessing the soundness of risk allocations in stress testing applications, and eventually for constructing more sound risk allocations based on multiple  scenarios as we will consider in Section~\ref{subsec:multimodality:adjustment}.}


\subsection{Multimodality adjustment of risk allocations}\label{subsec:multimodality:adjustment}

Multiple local modes of $\bX\ | \ \{S=K\}$ can be discovered in the process of estimating MLAs.
In this section, we discuss how to utilize the local modes, and introduce the so-called multimodality adjustment to increase the soundness of risk allocations under multimodality of $\bX\ | \ \{S=K\}$.

\subsubsection{Definition of multimodality adjustment}\label{subsubsec:definition:multimodality:adjustment}

Suppose that $M\in \IN$ scenarios $\sqc{\bK}{M} \in \mathcal K_d(K)$ are found with corresponding probability weights $\sqc{w}{M} \in [0,1]$ such that $\sum_{m=1}^M w_m=1$.
\rev{A typical choice of the scenario set $\{\sqc{\bK}{M}\}$ is the set of local modes of $\bx \mapsto f_{\bX}(\bx)\bone_{\{\bx \in \mathcal K_d(K)\}}$ (assumed to be a finite set), or possibly those belonging to a superlevel set at a certain level of plausibility.
The probability weight $w_m$ typically represents the likelihood of the scenario $\bK_m$ to occur, for instance, $w_m \propto f_{\bX}(\bK_m)/f_S(K)$.
Moreover, experts' assessments of the impact of the loss $\bK_m$ on the portfolio $\bX$ can also be incorporated.}
Multimodality adjustment is then defined as follows.

\begin{definition}[Multimodality adjustment of risk allocations]\label{def:multimodality:adjustment:risk:allocation}
Let $M\in \IN$ be the number of scenarios, $\mathcal X =\{\bK_1,\dots,\bK_M\}$ be the set of scenarios where $\bK_m \neq \bK_{m'}$ for any $m,\ m' \in \{1,\dots,M\}$ such that $m\neq m'$, and $\bw=(\sqc{w}{M})$ be the associated probability weights such that $\sum_{m=1}^M w_m=1$.
Then the \emph{multimodality-adjusted allocated capital} is defined by
\begin{align}\label{eq:multimodality:adjusted:risk:allocation}
\bK_{\bw,\mathcal X, \Lambda} =\bar \bK_{\bw, \mathcal X} + \sum_{m=1}^M w_m \blambda_m \circ (\bK_m-\bar \bK_{\bw, \mathcal X} )^{+},
\end{align}
where $\bar \bK_{\bw, \mathcal X}=\sum_{m=1}^M w_m \bK_m$ is the \emph{baseline allocation}, $\Lambda=(\blambda_1,\dots,\blambda_M) \in \IR_{+}^{d \times M}$ is the matrix of \emph{multimodality loading parameters}, $\bx\circ \by = (x_1y_1,\dots, x_d y_d)$ for $\bx,\ \by \in \IR^d$ and $\bx^{+}=(\max(x_1,0),\dots,\max(x_d,0))$ for $\bx \in \IR^d$.
We call the second term $\sum_{m=1}^M w_m \blambda_m \circ (\bK_m-\bar \bK_{\bw, \mathcal X} )^{+}$ of~\eqref{eq:multimodality:adjusted:risk:allocation} the \emph{multimodality adjustment}.
\end{definition}

\rev{Unlike MLA, the multimodality-adjusted allocated capital can be well-defined even if the global mode of $\bX\ | \ \{S=K\}$ is not unique.}
The capital~\eqref{eq:multimodality:adjusted:risk:allocation} consists of the baseline allocation and the additional loading to cover the variability of scenarios.
The baseline allocation is understood as an allocated capital before adjustment of multimodality.
Therefore, $\bar \bK_{\bw, \mathcal X}$ in~\eqref{eq:multimodality:adjusted:risk:allocation} can be replaced by Euler allocation if one requires its economic justification such as RORAC compatibility and core-compatibility.
To explain the multimodality adjustment, suppose that the scenario $\{\bX=\bK_m\}$ occurs with probability $w_m$.
Under this scenario, the portfolio incurs the loss (or profit) $\bK_m-\bar \bK_{\bw, \mathcal X}$.
When $\blambda_m=\bone_d$, the actual amount of loss $(\bK_m-\bar \bK_{\bw, \mathcal X})^{+}$ contributes to the multimodality adjustment $\sum_{m=1}^M w_m \blambda_m \circ (\bK_m-\bar \bK_{\bw, \mathcal X} )^{+}$.
However, this choice of $\blambda_m$ is too conservative, and smaller values of $\blambda_m$ are typically more reasonable since both $\bK_m$ and $\bar \bK_{\bw, \mathcal X}$ sum up to $K$ and thus losses of some units imply profits of others.
Therefore, losses of some units can be compensated by the profits of other units, and the multimodality loading parameter $\blambda_m$ can be determined by such risk mitigation or a corresponding insurance contract.

\subsubsection{Properties of the multimodality adjustment}\label{subsubsec:properties:multimodality:adjustment}

Next we study properties of $\bK_{\bw,\mathcal X, \Lambda}$. 
For certain choices of $\bw$ and $\mathcal X$, the capital $\bK_{\bw,\mathcal X, \Lambda}$ can be shown to satisfy translation invariance, positive homogeneity, riskless asset and symmetry; see Appendix~\ref{app:sec:further:properties:multimodality:adjustment} for details.
In the remainder of this section, we will verify that $\bK_{\bw,\mathcal X, \Lambda}$ measures the risk of multimodality from various viewpoints.
First, if $M=1$, then $\bK_{\bw,\mathcal X, \Lambda} =\bar \bK_{\bw, \mathcal X}$ and thus the multimodality adjustment is zero.
Second, suppose that $M\geq 2$ and $w_m>0$ for $m=1,\dots,M$. Then $\bK_{\bw,\mathcal X, \Lambda} =\bar \bK_{\bw, \mathcal X}$ if and only if
$\lambda_{j,m}=0$ for all $j=1,\dots,d$ and $m=1,\dots,M$ such that $K_{m,j}>\bar K_{\bw,\mathcal X,j}$.
Therefore, under multimodality, the multimodality adjustment is zero if and only if losses of some units of the portfolio are completely compensated by profits of others.
Finally, $\bK_{\bw,\mathcal X, \Lambda}$ is increasing with respect to the variability of the set of scenarios, which can be understood as a degree of multimodality.
To see this, suppose that $\blambda_1=\cdots=\blambda_M=\blambda$ for some $\blambda \in \IR_{+}^d$, and denote by $\bY$ the discrete random vector taking points $\bK_1,\dots,\bK_M$ with probabilities $\sqc{w}{M}$.
Then the multimodality-adjusted allocated capital~\eqref{eq:multimodality:adjusted:risk:allocation} can be written as
\begin{align}\label{eq:multimodality:adjusted:risk:allocation:stop:loss:representation}
\bK_{\bw,\mathcal X, \Lambda}= \E[\bY] + \blambda\circ \E[(\bY-\E[\bY])^{+}].
\end{align}
Variability of the set of scenarios can then be compared by the so-called convex order of $\sqc{Y}{d}$.
For two $\IR$-valued random variables $Y$ and $Y'$, $Y'$ is said to be larger than $Y$ in the \emph{convex order}, denoted as $Y \leq_{\text{cx}}Y'$, if $\E[\phi(Y)]\leq \E[\phi(Y')]$ for all convex functions $\phi:\IR\rightarrow \IR$ provided the expectations exist; see \citet{shaked2007stochastic} for a comprehensive reference.
Roughly speaking, convex order compares the variability of random variables and $Y'$ shows more variability than $Y$ if $Y \leq_{\text{cx}}Y'$;
for instance, $Y \leq_{\text{cx}}Y'$ implies $\E[Y]=\E[Y']$, $\Var(Y)\leq \Var(Y')$, $\text{ess.inf}(Y')\leq \text{ess.inf}(Y)$, $\text{ess.sup}(Y)\leq \text{ess.sup}(Y')$ and $\E[(Y-a)^{+}]\leq \E[(Y'-a)^{+}]$ for all $a \in \IR$.
Therefore, for two sets of scenarios $\mathcal X$ and $\mathcal X'$ with associated probabilities $\bw$ and $\bw'$, if one shows more variability than the other in the sense that the corresponding discrete random variables satisfy $Y_j\leq_{\text{cx}}Y_j'$ for some $j \in \{1,\dots,d\}$, then it holds that $(\bK_{\bw,\mathcal X, \Lambda})_j\leq (\bK_{\bw',\mathcal X', \Lambda})_j$.

\rev{
\begin{remark}[Multimodality adjustment for general sets of scenarios]
Representation~\eqref{eq:multimodality:adjusted:risk:allocation:stop:loss:representation} bears structural resemblance to Gini shortfall allocations introduced in \citet{furman2017gini}, and~\eqref{eq:multimodality:adjusted:risk:allocation:stop:loss:representation} indicates a possible extension of the multimodality adjustment to the case when the set of scenarios is not finite.
For instance, by taking $\mathcal X = \{\bx \in \IR^d: \bone_d\T\bx \geq \VaR{p}{S}\}$ and $w(\bx)=f_{\bX|\{S\geq \VaR{p}{S}\}}(\bx)$,~\eqref{eq:multimodality:adjusted:risk:allocation:stop:loss:representation} can be interpreted as multimodality-adjusted Euler allocations of Expected Shortfall since~\eqref{eq:multimodality:adjusted:risk:allocation:stop:loss:representation} yields
\begin{align*}
\bK_{\bw,\mathcal X, \Lambda}=\ES{p}{X_j;S}+ \blambda \circ \E[(\bX-\ES{p}{X_j;S})^{+}\ | \ \{S\geq \VaR{p}{S}\}],
\end{align*}
where $\ES{p}{X_j;S}=\E[\bX\ | \ \{S\geq \VaR{p}{S}\}]$ is the Euler allocation of $K=\ES{p}{S}$ as derived in~\eqref{eq:ES:contributions}.
\end{remark}
}


We end this section with a remark on the case when multiple measures or models are considered as different scenarios and how to incorporate these scenarios into multimodality-adjusted allocated capitals.

\begin{remark}[Multimodality adjustment for different measures]
A single model of a risk may not be sufficient to manage the risk due to changes of an economic situation or due to model uncertainty.
For a further risk assessment, it may be useful to consider multiple measures
$\bQ_1,\dots,\bQ_S$ where $\bQ_s$ is a probability measure on $(\Omega,\mathcal A)$ and $F_{\bX}^{\bQ_s}$ is the distribution function of $\bX$ under $\bQ_s$ for $s=1,\dots,S$.
These multiple measures can be incorporated into the scenario analysis by, for example, considering the (componentwise) maximum of the multimodality-adjusted allocated capitals $\mathcal K_{\bw,\mathcal X,\Lambda}$  calculated based on $F_{\bX}^{\bQ_s}$ for $s=1,\dots,S$, or considering their mixture with respect to probabilities $q_1,\dots,q_S$ where $q_s$ is associated to the scenario $\bQ_s$ determined, for example, proportionally to the sample size available for the distribution $F_{\bX}^{\bQ_s}$.
\end{remark}


\section{Numerical experiments}\label{sec:numerical:experiments}
In this section we conduct an empirical and simulation study to compute Euler
and maximum likelihood allocations, and compare them for various models.
Simulation of the conditional distribution given a constant sum is in general
challenging.  Throughout this section, we adopt \emph{(crude)
  Monte Carlo (MC)} method to simulate $\bX'\ | \ \{S=K\}$ according to which
unconditional samples from $\bX$ are first generated and those falling in the
region
$\mathcal K_d(K,\delta)=\{\bx \in \IR^d: K-\delta<\bone_d\T\bx<K+\delta\}$ for a
sufficiently small $\delta>0$ are then extracted. The extracted samples are
standardized via $K X_{j}/\sum_{j=1}^d X_{j}$ so that their componentwise sum
equals $K$.  Finally the standardized samples are used as pseudo-samples from
$\bX'\ | \ \{S=K\}$.  See Appendix~\ref{app:subsec:simulation:HMC} for the
potential bias produced by this method, and more sophisticated simulation
approaches of $\bX'\ | \ \{S=K\}$ based on MCMC methods.
\rev{All experiments are run on a MacBook Air with 1.4 GHz Intel Core i5
processor and 4 GB 1600 MHz of DDR3 RAM.}

\subsection{Empirical study}\label{subsec:empirical:study}
In this section we estimate MLA nonparametrically for real financial data.
We consider daily log-returns of the stock indices FTSE $X_{t,1}$, S\&P 500 $X_{t,2}$ and DJI $X_{t,3}$ from January 2, 1990 to March 25, 2004, which contains $3713$ days and thus $T=3712$ log-returns.
We consider two portfolios (a) $\bX^{\text{pos}}_t=(X_{t,1},X_{t,2},X_{t,3})$ and (b) $\bX_t^{\text{neg}}=(X_{t,1},-X_{t,2},X_{t,3})$.
For each portfolio, we aim at allocating the capital $K=1$ based on the conditional loss distribution at time $T+1$ given the history up to and including time $T$.
Taking into account the stylized facts of stock returns listed in Chapter 3 of \citet{mcneil2015quantitative} (such as unimodality, heavy-tailedness and volatility clusters), we adopt a copula-GARCH model with marginal skew-$t$ innovations (ST-GARCH; see, for example, \citet{jondeau2006copula} and \citet{huang2009estimating}).
We utilize a GARCH$(1,1)$ model with skew-$t$ innovations with degrees of freedom $\nu_j>0$ and skewness parameter $\gamma_j>0$ for the $j$th marginal time series.
That is, within a fixed time period $\{1,\dots,T+1\}$, the $j$th return series $(X_{1,j},\dots,X_{T+1,j})$ follows
\begin{align*}
X_{t,j}=\mu_j + \sigma_{t,j}Z_{t,j},\quad \sigma_{t,j}^2 = \omega_{j}+\alpha_{j}X_{t-1,j}^2+\beta_{j}\sigma_{t-1,j}^2,\quad Z_{t,j}\iidsim \text{ST}(\nu_j,\gamma_j),\quad j=1,\dots,d,
\end{align*}
where $\omega_{j} >0,\alpha_{j},\beta_{j}\geq 0$, $\alpha_{j}+\beta_{j}<1$, and $Z_{t,j}$ follows a skew-$t$ distribution $\text{ST}(\nu_j,\gamma_j)$ with density given by
\begin{align}\label{eq:density:skew:t:distribution}
f_{j}(x_j; \nu_j,\gamma_j)=\frac{2}{\gamma_j+\frac{1}{\gamma_j}}\left\{t(x_j,\nu_j)1_{[x_j\geq 0]} + t(\gamma_j x_j,\nu_j)1_{[x_j< 0]}\right\},
\end{align}
where $t(x,\nu)$ is the density function of a Student $t$ distribution with degrees of freedom $\nu>0$ and a skewness parameter $\gamma >0$ with $\gamma=1$ corresponding to the standard symmetric case; see \citet{fernandez1998bayesian} for more details.
The copula among the stationary process $\bZ_t = (Z_{t,1},\dots,Z_{t,d})$, denoted as $C$, is estimated nonparametrically.
Under this model, the joint distribution of the returns $\bX_{T+1|\mathcal F_T}=(X_{T+1,1|\mathcal F_{T}},\dots,X_{T+1,d|\mathcal F_{T}})$ has marginal distributions ST$(\mu_j,\sigma_{t+1,j}^2,\nu_j,\gamma_j)$, $j=1,\dots,d$, and a copula $C$, where ST$(\mu_j,\sigma_{t+1,j}^2,\nu_j,\gamma_j)$ is a skew-$t$ distribution with density $f_{j}(\frac{x_j-\mu_j}{\sigma_{t+1,j}}; \nu_j,\gamma_j)$ with $f_{j}(\cdot;\nu_j,\gamma_j)$ defined in \eqref{eq:density:skew:t:distribution}.
Parameters of the ST-GARCH(1,1) models are estimated with the maximum likelihood method; the results are summarized in Table~\ref{table:parameters:estimators:ST:GARCH}.

\begin{table}[t]
  \caption{Maximum likelihood estimates and estimated standard errors of the ST-GARCH(1,1) model $X_{t,j}=\mu_j + \sigma_{t,j}Z_{t,j}$ with $\sigma_{t,j}^2 = \omega_{j}+\alpha_{j}X_{t-1,j}^2+\beta_{j}\sigma_{t-1,j}^2$ and $Z_{t,j}\iidsim \text{ST}(\nu_j,\gamma_j)$ for $j=1,2,3$.
  }\label{table:parameters:estimators:ST:GARCH}
  \centering
  \begin{tabular}{l
    S[table-format=1.3, table-sign-mantissa]
    S[table-format=1.3]
    S[table-format=1.3]
    S[table-format=1.3]
    S[table-format=1.3]
    S[table-format=1.3]
    }
    \toprule
    & \multicolumn{1}{c}{$\mu_j$} & \multicolumn{1}{c}{$\omega_j$} & \multicolumn{1}{c}{$\alpha_j$} & \multicolumn{1}{c}{$\beta_j$} & \multicolumn{1}{c}{$\gamma_j$} & \multicolumn{1}{c}{$\nu_j$}\\
    \midrule
    $X_{t,1}^\text{pos}=X_{t,1}^\text{neg}$ & 0.053 & 0.006 & 0.052 & 0.943 & 0.969 & 6.414 \\
    SE & 0.013 & 0.002 & 0.008 & 0.008 & 0.021 & 0.663 \\[3mm]
    $X_{t,2}^\text{pos}$ & 0.050 & 0.003 & 0.049 & 0.950 & 0.983 & 6.265 \\
    SE & 0.013 & 0.001 & 0.007 & 0.007 & 0.021 & 0.659 \\[3mm]
    $X_{t,2}^\text{neg}$ & -0.050 & 0.003 & 0.049 & 0.950 & 1.018 & 6.265 \\
    SE & 0.013 & 0.001 & 0.007 & 0.007 & 0.022 & 0.659 \\[3mm]
    $X_{t,3}^\text{pos}=X_{t,3}^\text{neg}$ & 0.031 & 0.011 & 0.071 & 0.920 & 0.966 & 10.000 \\
    SE & 0.014 & 0.003 & 0.009 & 0.010 & 0.023 & 1.309 \\
    \bottomrule
  \end{tabular}
\end{table}

For each case of (a) $\bX^\text{pos}$ and (b) $\bX^\text{neg}$, we estimate the Euler allocation and MLA by a resampling method.
After extracting the marginal standardized residuals, we build their pseudo-observations as a pseudo-sample from $C$.
We then generate samples of size $N=3712$ by resampling with replacement.
The samples from $C$ are then marginally transformed by skew-$t$ distributions with parameters specified as in Table~\ref{table:parameters:estimators:ST:GARCH}.
From these samples of $\bX_{T+1|\mathcal F_T}$, we extract the subsamples falling in the region $\mathcal K_d(K,\delta)= \left\{\bx \in \IR^3: K-\delta < \sum_{j=1}^3 x_j <  K + \delta \right\} $ where $\delta = 0.3$.
These samples are then standardized via $K X_{t,j}/\sum_{j=1}^d X_{t,j}$.
Scatter plots of the first two components of these samples are shown in Figure~\ref{fig:plots:empirical:plot}.

The 3712 data points lead to 354 and 558 samples from $\bX_{T+1|\mathcal F_T}^{\text{pos}}$ and $\bX_{T+1|\mathcal F_T}^{\text{neg}}$ on $\mathcal K_d(K,\delta)$, respectively.
Based on these conditional samples, we estimate the Euler allocation $\E[\bX\ |\ \{S=K\}]$ and the MLA, that is, the mode of $f_{\bX|\{S=K\}}$ provided it is unique.
\rev{
The (possibly multiple) modes were estimated by the function \textsf{kms} \citep[\emph{kernel mean shift clustering}, proposed by][]{fukunaga1975estimation} of the \textsf{R} package \textsf{ks}; see \citet{carreira2015review} and \citet{chen2016comprehensive} for details and for other methods of estimating modes.}
\rev{For the computational times, computing MLAs took 0.35 
seconds in Case (a) and 0.43 
seconds in Case (b) whereas, in both cases, the Euler allocations were computed almost instantly.}
As was expected from the ellipticality of the scatter plots in Figure~\ref{fig:plots:empirical:plot}, the unique mode was discovered in each case.
The first two components of the two allocations are pointed out in Figure~\ref{fig:plots:empirical:plot}.

\begin{figure}[t]
  \centering
  \vspace{-35mm}
  \includegraphics[width=15 cm]{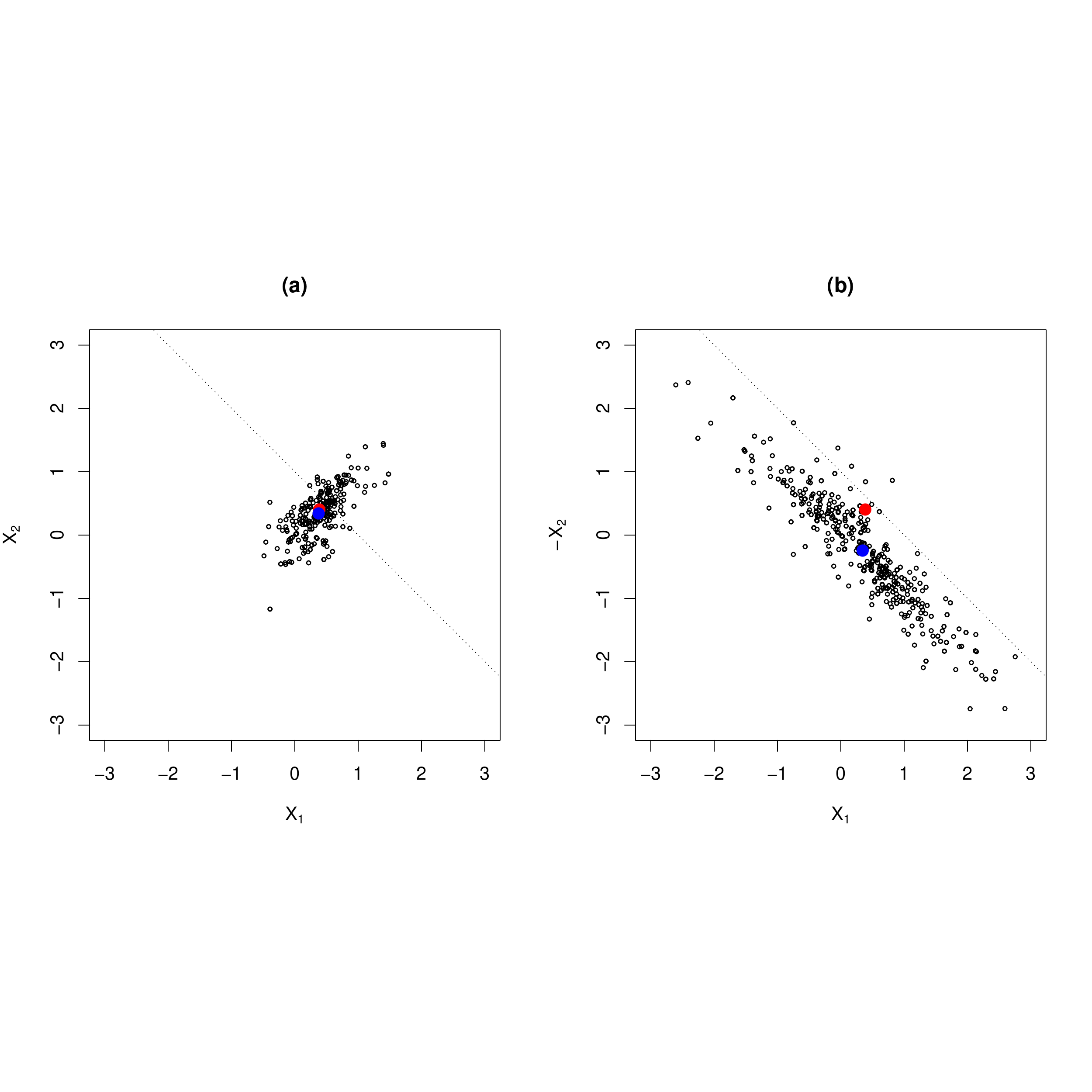}\vspace{-35mm}
  \caption{Scatter plots (black dots) of the first two components of (a) $\bX_{T+1|\mathcal F_T}^{\text{pos}}=(X_{T+1,1|\mathcal F_T},X_{T+1,2|\mathcal F_T},X_{T+1,3|\mathcal F_T})$ and (b) $\bX_{T+1|\mathcal F_T}^{\text{neg}}=(X_{T+1,1|\mathcal F_T},-X_{T+1,2|\mathcal F_T},X_{T+1,3|\mathcal F_T})$ for daily log-returns of the stock indices FTSE $X_{t,1}$, S\&P 500 $X_{t,2}$ and DJI $X_{t,3}$ falling in the region $\mathcal K_d(K,\delta)=\left\{\bx \in \IR^3: K-\delta < \sum_{j=1}^3 x_j <  K + \delta \right\} $ where $\delta = 0.3$ and $K=1$.
  The dotted lines represent $x + y =K$.
  The red dot represents the Euler allocation $\E[\bX'\ |\ \{S=K\}]$ and the blue dot represents the maximum likelihood allocation, which is the mode of  $f_{\bX'|\{S=K\}}$.
  }
  \label{fig:plots:empirical:plot}
\end{figure}

Next, we estimate the standard errors of the Euler and maximum likelihood allocations using the bootstrap method.
We compute the Euler allocation, MLA and their standard errors based on the $B=100$ number of samples of size $N=3712$ resampled from the original data with replacement.
The results are summarized in Table~\ref{table:estimators:empirical:study}.

\begin{table}[t]
\centering
\caption{Bootstrap estimates and estimated standard errors of the Euler allocation and MLA of $\bX_{T+1|\mathcal F_T}^{\text{pos}}=(X_{T+1,1|\mathcal F_T},X_{T+1,2|\mathcal F_T},X_{T+1,3|\mathcal F_T})$ and $\bX^{\text{neg}}=(X_{T+1,1|\mathcal F_T},-X_{T+1,2|\mathcal F_T},X_{T+1,3|\mathcal F_T})$ for daily log-returns of the stock indices FTSE $X_{t,1}$, S\&P 500 $X_{t,2}$ and DJI $X_{t,3}$.
The subsample size is $N=3712$ and the bootstrap sample size is $B=100$.}\label{table:estimators:empirical:study}\vspace{5mm}
\begin{tabular}{l rrr rrr}
  \toprule
  & \multicolumn{3}{c}{Estimator} & \multicolumn{3}{c}{Standard error}\\
  \cmidrule(lr{0.4em}){2-4}\cmidrule(lr{0.4em}){5-7}
  & \multicolumn{1}{c}{$X_1$} & \multicolumn{1}{c}{$X_2$} & \multicolumn{1}{c}{$X_3$} & \multicolumn{1}{c}{$X_1$} & \multicolumn{1}{c}{$X_2$} & \multicolumn{1}{c}{$X_3$}\\
  \midrule
  $\E[\bX^{\text{pos}}\ |\ \{S=K\}]$ & 0.378 & 0.338 & 0.285 & 0.019 & 0.022 & 0.038\\
  $\bK_{\operatorname{M}}[\bX^{\text{pos}};\mathcal K_d(K)]$ & 0.367 & 0.365 & 0.268 & 0.019 & 0.024 & 0.041\\[3mm]
  $\E[\bX^{\text{neg}}\ |\ \{S=K\}]$ & 0.345 & $-$0.248 & 0.903 & 0.037 & 0.039 & 0.015 \\
  $\bK_{\operatorname{M}}[\bX^{\text{neg}};\mathcal K_d(K)]$ & 0.371 & $-$0.280 & 0.909 & 0.040 & 0.039 & 0.013\\
  \bottomrule
\end{tabular}
\end{table}

In Figure~\ref{fig:plots:empirical:plot} we can observe that compared with Case
(a) the distribution in Case (b) is more spread out and losses take larger
absolute values. If the samples are regarded as stressed
scenarios, the scenario set in Case (b) contains a wider variety of
scenarios than in Case (a) since both of positive and negative losses can
appear in Case (b) whereas most realizations are positive in Case (a).
Nevertheless, as is observed from Table~\ref{table:estimators:empirical:study},
in both cases the Euler allocation and the MLA are close to each other also in
terms of standard errors.  This observation does not conflict with the stylized
fact that the joint log-returns nearly follow an elliptical distribution, and
thus the mean (Euler allocation) of $\bX\ | \ \{S=K\}$ coincides with its mode;
see Proposition~\ref{prop:ellipticality:conditional:distribution:constant:sum}
and Proposition~\ref{prop:unimodality:conditional:distribution:constant:sum}
Part 1.


\subsection{Simulation study}\label{subsec:simulation:study}

In this section, we consider four
models, referred to as (M1), (M2), (M3) and (M4), respectively, with $d=3$ and having the same marginal distributions
$X_1 \sim \Par(2.5,5)$, $X_2 \sim \Par(2.75,5)$ and $X_3 \sim \Par(3,5)$ (where $\Par(\theta,\lambda)$ denotes the Pareto distribution with shape parameter $\theta>0$ and scale parameter $\lambda>0$) but different $t$ copulas with degrees of freedom $\nu=5$ and dispersion matrices
\begin{align}\label{eq:parameters:correlation:matrices}
  \nonumber P_1 &= \begin{pmatrix}
    1& 0.8 & 0.5\\
    0.8 & 1 & 0.8\\
    0.5 & 0.8 & 1\\
  \end{pmatrix},\quad
  P_2 = \begin{pmatrix}
    1 & 0.5 & 0.5\\
    0.5 & 1 & 0.5\\
    0.5 & 0.5 & 1\\
  \end{pmatrix}, \\
  P_3 &= \begin{pmatrix}
    1& 0 & 0.5\\
    0 & 1 & 0 \\
    0.5 & 0 & 1 \\
  \end{pmatrix},\quad
  P_4 = \begin{pmatrix}
    1& -0.5 &  0.5\\
    -0.5 & 1 & -0.5\\
    0.5 & -0.5 & 1\\
  \end{pmatrix},
\end{align}
respectively.
For these parametric models, we first simulate $N=10^6$ samples from the unconditional distribution and then extract subsamples falling in the region $\mathcal K_d(K,\delta)$ with
$K=40$ and $\delta = 1$.
These pseudo samples from $\bX'\ | \ \{S=K\}$ are shown in Figure~\ref{fig:plots:simulation:plot}.
The red point in the figure represents the Euler allocation and the blue points are the (local) modes, which are estimated similarly as in Section~\ref{subsec:empirical:study}.

\begin{figure}
  \centering
  \vspace{-5mm}
  \includegraphics[width=15cm]{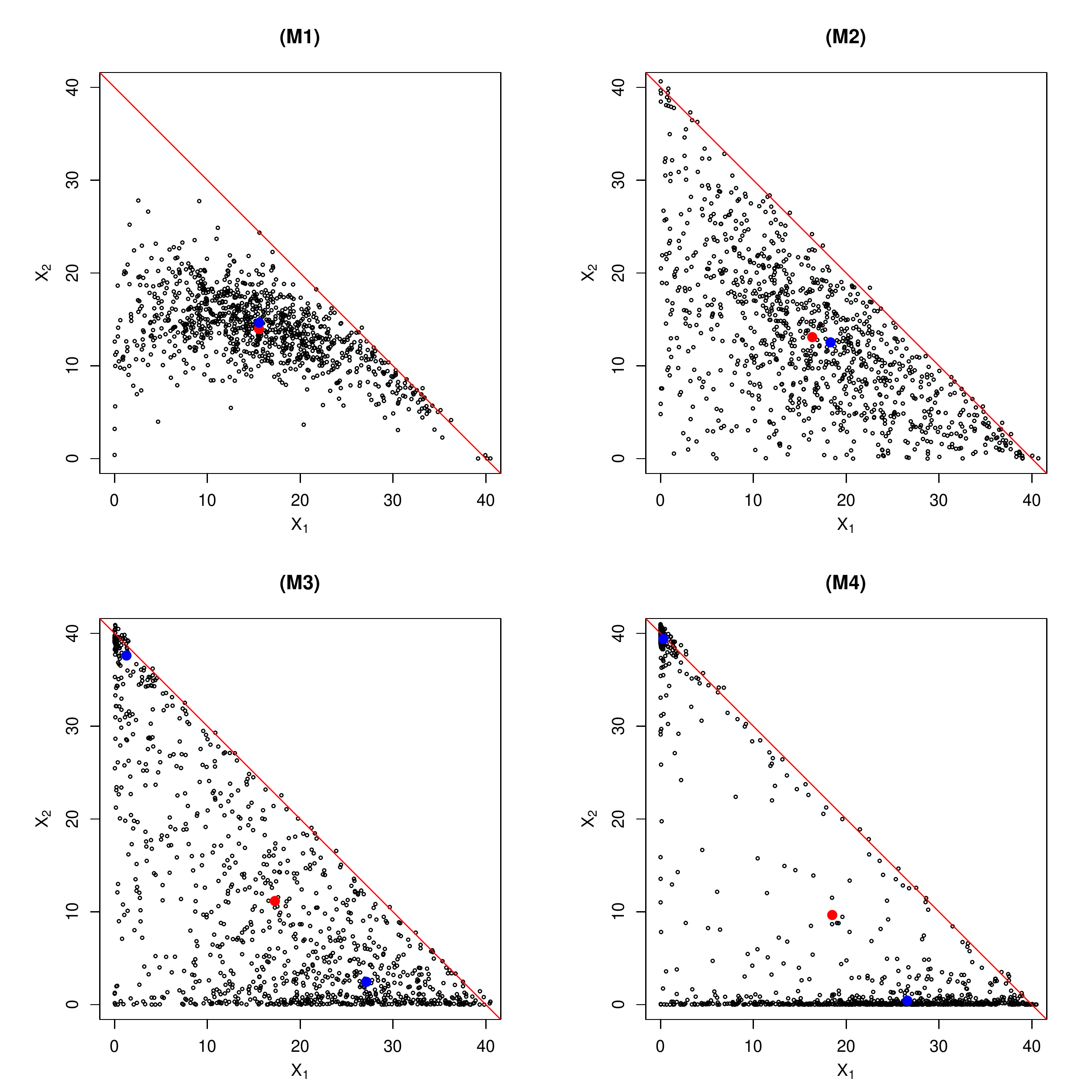}\vspace{-5mm}
  \caption{Scatter plots (black dots) of the first two components of the four models (M1), (M2), (M3) and (M4) falling in the region $\mathcal K_d(K,\delta)$
  with $K=40$ and $\delta = 1$.
  All the four models have the same marginal distributions $X_1 \sim \Par(2.5,5)$, $X_2 \sim \Par(2.75,5)$ and $X_3 \sim \Par(3,5)$ but different $t$ copulas with parameters provided in~\eqref{eq:parameters:correlation:matrices}.
  The red lines represent $x + y =K$.
  The red dot represents the Euler allocation $\E[\bX'\ |\ \{S=K\}]$ and the blue dots represent the (local) modes of $f_{\bX'|\{S=K\}}$.
  }
  \label{fig:plots:simulation:plot}
\end{figure}

\rev{The computational times required for calculating MLAs were (in seconds) (M1) 18.96
, (M2) 11.72
, (M3) 15.94
, and (M4) 22.76 
whereas the Euler allocations were computed almost instantly for all the cases.
Compared with the results in Section~\ref{subsec:empirical:study}, we observe that the computational time required to calculate MLA increases more rapidly than the Euler allocation does as the sample size increases.
}

In Figure~\ref{fig:plots:simulation:plot} we can observe that the conditional distribution is more concentrated under positive dependence (Model (M1) and (M2)) and it is more dispersed under negative dependence (Model (M4)).
Regarding the samples as stressed scenarios, the sets for Model (M3) and (M4) are more worrisome than those for Model (M1) and (M2) since (M3) and (M4) contain two distinct scenarios, one around the first axis and another around the upper-left corner of the plot region, both of which are likely to occur in the stressed situation $\{S=K\}$.
Unimodality of the conditional distribution for Model (M1) and (M2) leads to closer Euler allocation and MLA.
For Model (M1) and (M2), the choice of Euler allocation and MLA does not significantly change the resulting allocation.
On the other hand, for Model (M3) and (M4), the conditional distributions are multimodal, \rev{and thus a more careful decision making may be required}.

\begin{table}[t]
\centering
\caption{Estimates and estimated standard errors of the Euler allocation and MLA
of the four models (M1), (M2), (M3) and (M4) all having the same marginal distributions $X_1 \sim \Par(2.5,5)$, $X_2 \sim \Par(2.75,5)$ and $X_3 \sim \Par(3,5)$ but different $t$ copulas with parameters provided in~\eqref{eq:parameters:correlation:matrices}.
Estimates and estimated standard errors are computed based on $100$ replications, each of which consists of 500 conditional samples falling in the region $\mathcal K_d(K,\delta)$ with $K=40$ and $\delta=1$.
}\label{table:estimators:simulation:study}\vspace{5mm}
\begin{tabular}{l rrr rrr}
  \toprule
  & \multicolumn{3}{c}{Estimator} & \multicolumn{3}{c}{Standard error}\\
  \cmidrule(lr{0.4em}){2-4}\cmidrule(lr{0.4em}){5-7}
  & \multicolumn{1}{c}{$X_1$} & \multicolumn{1}{c}{$X_2$} & \multicolumn{1}{c}{$X_3$} & \multicolumn{1}{c}{$X_1$} & \multicolumn{1}{c}{$X_2$} & \multicolumn{1}{c}{$X_3$}\\
  \midrule
    \multicolumn{7}{l}{(M1) Pareto + $t$ copula: strong positive dependence}\\[5pt]
 $\E[\bX\ |\ \{S=K\}]$ & 15.549 & 13.889 & 10.562  & 0.336 & 0.157 & 0.288 \\
 $\bK_{\operatorname{M}}[\bX;\mathcal K_d(K)]$ & 15.849 & 14.434 & 9.718 & 0.482 & 0.213 & 0.356\\[3mm]

    \multicolumn{7}{l}{(M2) Pareto + $t$ copula: positive dependence}\\[5pt]
 $\E[\bX\ |\ \{S=K\}]$ & 16.228 & 13.042 & 10.562 & 0.399 & 0.355 & 0.288\\
 $\bK_{\operatorname{M}}[\bX;\mathcal K_d(K)]$ & 17.689 & 12.481 & 9.830 & 0.759 & 0.663 & 0.475\\[3mm]

    \multicolumn{7}{l}{(M3) Pareto + $t$ copula: no correlation}\\[5pt]
 $\E[\bX\ |\ \{S=K\}]$ & 17.479 & 11.368 & 10.562 & 0.517 & 0.530 & 0.288\\
 $\bK_{\operatorname{M},1}[\bX;\mathcal K_d(K)]$ & 25.678 & 3.107 & 11.215  & 1.185 & 0.278 & 1.205\\
 $\bK_{\operatorname{M},2}[\bX;\mathcal K_d(K)]$ &  2.639 & 35.275 & 2.086  & 0.973 & 1.306 & 0.424\\[3mm]

    \multicolumn{7}{l}{(M4) Pareto + $t$ copula: negative dependence}\\[5pt]
 $\E[\bX\ |\ \{S=K\}]$ & 19.062 & 9.272 & 10.562 &  0.556 & 0.614 & 0.288\\
 $\bK_{\operatorname{M},1}[\bX;\mathcal K_d(K)]$ &  28.353 & 0.684 & 10.962  & 2.125 & 1.646 & 2.154 \\
 $\bK_{\operatorname{M},2}[\bX;\mathcal K_d(K)]$ &  0.710 & 38.385 & 0.905  & 1.719 & 3.537 & 2.705 \\
 \bottomrule
\end{tabular}
\end{table}

To investigate the standard errors of the estimators, we compute the estimates of Euler allocation and (local) modes of $f_{\bX'|\{S=K\}}$ $100$ times for each model.
For each repetition, we simulate samples from $\bX$ so that there are $500$ samples in the region $\mathcal K_d(K,\delta)$.
The estimates and standard errors are computed based on the 100 replications and the results are summarized in Table~\ref{table:estimators:simulation:study}.
We can again see that for Models (M1) and (M2) the mode of $\bX'\ |\ \{S=K\}$ is unique and the two allocations are close.
On the other hand, for Models (M3) and (M4) where the conditional distributions are multimodal, the standard errors of the (local) modes are higher than those of the Euler allocation.


\rev{
In the end, we compute the multimodality-adjusted allocated capitals~\eqref{eq:multimodality:adjusted:risk:allocation} for (M3) and (M4).
In each case, the parameters are specified as $M=2$, $\mathcal X = \{\bK_1,\bK_2\}$ and $\bw=(w_1,w_2)$ with $w_m\propto f_{\bX}(\bK_m)$, where $\bK_1=\bK_{\operatorname{M},1}[\bX;\mathcal K_d(K)]=(26.726,2.114,11.158)$ and $\bK_2=\bK_{\operatorname{M},2}[\bX;\mathcal K_d(K)]=(1.505, 37.203, 1.291)$ for (M3), and $\bK_1 = (28.589, 0.432, 10.978)$ and $\bK_2=( 0.326, 39.314, 0.358)$ for (M4).
For both (M3) and (M4), the first and third units incur losses when $\{\bX=\bK_1\}$ occurs, and the second unit incurs a large loss when $\{\bX=\bK_2\}$ happens.
The probability weights of the scenarios are given by $\bw=(0.509,0.490)$ for (M3) and $\bw = (0.272,0.727)$ for (M4).
The two scenarios $\bK_1$ and $\bK_2$ are almost equally likely to occur for (M3), and the second scenario $\bK_2$ is more likely to occur for (M4).
Based on $\bw$ and $\mathcal X$, the baseline allocations $\bK_{\bw,\mathcal X}$ are given by $(14.357,19.323,6.319)$ for (M3) and $(8.038,28.705,3.256)$ for (M4).
As seen in Table~\ref{table:estimators:simulation:study}, these allocations are not quite close to the Euler allocations since $\bK_{\bw,\mathcal X}$ is calculated based only on the two points $\bK_1$ and $\bK_2$ in $\mathcal K_d(K)$.
}

\rev{
On computing the multimodality adjustments, we consider two cases when $\bK_{\bw,\mathcal X}$ or Euler allocations are used as baseline allocations.
If $\bK_{\bw,\mathcal X}$ is the baseline allocation, the average loss $w_1 (\bK_1-\bK_{\bw,\mathcal X})^{+}+w_2(\bK_2-\bK_{\bw,\mathcal X})^{+}$ in the multimodality adjustment is given by $(6.303, 11.204, 0.000)$ for (M3) and $(5.607,22.742,0.000)$ for (M4).
If Euler allocation is used as a baseline allocation, the average loss is given by $( 4.873, 9.828, 0.000)$ for (M3) and $(2.618, 14.775,  0.428)$ for (M4).
In all cases, the average loss incurred in the second unit is larger than those in the first and third units since the second unit incurs a large loss when the second scenario $\{\bX=\bK_2\}$ occurs.
Moreover, for (M4), the probability that the scenario $\{\bX=\bK_2\}$ occurs is higher than that of the first scenario $\{\bX=\bK_1\}$.
Therefore, the analysis of the modes of $\bX\ | \ \{S=K\}$ reveals that a large multimodality adjustment should be applied to $\bX_2$ to increase the soundness of risk allocations.
}

\section{Conclusion}\label{sec:discussion:conclusion}
\rev{Motivated from stress testing of risk allocations,}
we investigated properties of the conditional distribution of $\bX$
given the sum constraint $\{S=K\}$, and also introduced the novel risk allocation method called maximum
likelihood allocation (MLA). 
The \rev{super}level
set of $\bX\ | \ \{S=K\}$ can be regarded as a set of stressed (severe and
plausible) scenarios, and the modality of $\bX\ | \ \{S=K\}$ can be interpreted
as a variety of risky scenarios, which turned out to be an important
feature \rev{in risk assessment} related to the soundness of risk allocations. 
We then studied modality, dependence and tail behavior of $\bX\ | \ \{S=K\}$,
most of which are inherited from those of the unconditional loss $\bX$.  
\rev{We then investigated how to incorporate the knowledge of the modality of $\bX\ | \ \{S=K\}$ for more sound risk management.
Under unimodality, we
defined MLA as a mode of $\bX\ | \ \{S=K\}$, and studied its properties as a risk allocation principle, such as translation invariance and positive homogeneity.  
Under multimodality, we considered the so-called multimodality adjustment to increase the soundness of risk allocations based on the multiple modes.}
Euler allocation and MLA were then
compared in numerical experiments.  
Through the experiments, we demonstrated
that $\bX\ | \ \{S=K\}$ is
typically unimodal, and Euler allocation and MLA lead to close values when $\bX$ possesses positive dependence.  
On the other hand,
when the losses are negatively dependent, multimodality is likely to occur, and 
searching for the modes of $\bX\ | \ \{S=K\}$ is beneficial for discovering hidden risky scenarios,
evaluating the soundness of risk allocations, and  
\rev{eventually increasing the soundness of risk allocations by imposing the multimodality adjustment.}

Although we empirically observed the relationship between multimodality of
$\bX\ | \ \{S=K\}$ and negative dependence of $\bX$, this relationship
requires further theoretical investigation.  Another aspect of future research
is to study more distributional properties, such as tail dependence and measures
of concordance, of $\bX\ | \ \{S=K\}$ especially without assuming the existence
of a density.
\rev{
Unlike Euler allocations, estimation of MLAs is not a straightforward problem in general but various methods are known for estimating modes of multivariate distributions.
For applying the MLA principle in practice, efficient estimation methods of the modes of multivariate distributions in high dimensions need to be explored further.
}
\rev{An economic justification of the MLA principle is also an interesting direction for future research.}
\rev{
In addition, an extension of the multimodality adjustment to general sets of scenarios is another avenue to explore since the multimodality adjustment proposed in this paper relies on the assumption that the set of scenarios is finite.
}
In the end, efficient simulation approaches of $\bX\ | \ \{S=K\}$
may need to rely on MCMC methods as introduced in
Appendix~\ref{app:simulation:with:MCMC}, and further investigation is required
to assess how beneficial the distributional properties proven in this paper are to construct efficient MCMC methods since the performance of MCMC methods typically depends on
tail-heaviness and modality of the target distribution.
\section*{Funding}
This research was funded by NSERC through Discovery Grant RGPIN-5010-2015.

\section*{Declaration of interests}
Declarations of interest: none.

\section*{Acknowledgments}
\rev{We would like to thank the editor and three anonymous referees for their insightful comments on our manuscript.
}

\bibliographystyle{apalike}

\appendix

\section*{Appendices}

\section{Proofs}\label{app:proofs}

\subsection*{Proof of Proposition~\ref{prop:ellipticality:conditional:distribution:constant:sum}}
\begin{proof}
Notice that $(\bX',S)=A\bX \sim {\mathcal E}_d(A\bmu,A\Sigma A\T,\psi)$ where
$A=\begin{pmatrix}
\bI_d & \bzero_d\\
\bone_d\T & 1\\
\end{pmatrix} \in \IR^{d\times d}$.
Therefore, the conditional distribution $\bX'\ |\ \{S=K\}$ also follows an elliptical distribution with the location parameter $\bmu_K$ and the dispersion parameter $\Sigma_K$ as specified in \eqref{eq:location:dispersion:parameter:conditional:dist:constant:sum}.
The corresponding characteristic generator $\psi_K$ can be specified through Theorem 2.18 of \citet{fang2018symmetric}.
If $\bX$ admits a density with density generator $g$, then
\begin{align*}
f_{\bX'|\{S=K\}}(\bx') &= \frac{f_{(\bX',S)}(\bx',K)}{f_S(K)}
\propto g_{d}\left(
\frac{1}{2} (\bx'-\bmu',K - \mu_S)  \T
\begin{pmatrix}
\Sigma' & (\Sigma\bone_d)'\\
\tra{(\Sigma\bone_d)'} & \sigma_S^2\\
\end{pmatrix}^{-1}
 (\bx'-\bmu',K - \mu_S)
\right).
\end{align*}
The quadratic term reduces to
\begin{align*}
(\bx'-\bmu',K - \mu_S)  \T
&\begin{pmatrix}
\Sigma' & (\Sigma\bone_d)'\\
\tra{(\Sigma\bone_d)'} & \sigma_S^2\\
\end{pmatrix}^{-1}
 (\bx'-\bmu',K - \mu_S)=(\bx'-\bmu_K)\T \Sigma_K^{-1} (\bx'-\bmu_K) +\left(\frac{K-\mu_S}{\sigma_S}\right)^2.
\end{align*}
Therefore, we have that
\begin{align*}
f_{\bX'|\{S=K\}}(\bx') &\propto g\left(
\frac{1}{2}
(\bx'-\bmu_K)\T \Sigma_K^{-1} (\bx'-\bmu_K) + \Delta_K
\right)
 = g_K\left(
\frac{1}{2}
(\bx'-\bmu_K)\T \Sigma_K^{-1} (\bx'-\bmu_K)
\right),
\end{align*}
where $\Delta_K =  (K-\mu_S)^2/(2\sigma_S^2)$ and $g_K(t)=g(t  + \Delta_K)$ as specified in \eqref{eq:density:generator:conditional:distribution:constant:sum}.
\end{proof}

\subsection*{Proof of Proposition~\ref{prop:unimodality:conditional:distribution:constant:sum}}

\begin{proof}
\begin{enumerate}
\item By Proposition~\ref{prop:ellipticality:conditional:distribution:constant:sum}, $\bX'\ |\ \{S=K\}$ follows a $d'$-dimensional elliptical distribution with location vector $\bmu_K$, dispersion matrix $\Sigma_K$ and density generator $g_K$.
Furthermore, $g_K$ is decreasing if $g$ is.
Therefore, for $0<s\leq c_K t_K^\ast/\sqrt{|\Sigma_K|}$,
\begin{align*}
L_s(\bX'\ |\ \{S=K\}) &= \left\{\bx'\in\IR^{d'}: g_K\left(\frac{1}{2}(\bx'-\bmu_K)\T\Sigma_K^{-1}(\bx'-\bmu_K)\right) \geq \frac{s \sqrt{|\Sigma_K|}}{c_K}\right\}\\
&=  \left\{\bx'\in\IR^{d'}: 0 \leq (\bx'-\bmu_K)\T\Sigma_K^{-1}(\bx'-\bmu_K) \leq
2 \left\{g^{-1}\left(\frac{s \sqrt{|\Sigma_K|}}{c_K}\right)-\Delta_K\right\}\right\},
\end{align*}
which is a convex set with ellipsoid as surface.
Moreover, when $s^\ast = c_K t_K^\ast/\sqrt{|\Sigma_K|}$, we have
\begin{align*}
L_{s^{\ast}}(\bX'\ |\ \{S=K\})= \left\{\bx'\in\IR^{d'}:  (\bx'-\bmu_K)\T\Sigma_K^{-1}(\bx'-\bmu_K) =0\right\}=\{\bmu_K\}
\end{align*}
and thus $\bX'\ |\ \{S=K\}$ has a mode $\bmu_K$.
\item
For $t>0$ and $\bx' \in \IR^{d'}$, we have the equivalence relation:
\begin{align}\label{eq:evuivalence:level:sets}
\bx' \in L_t(\bX'\ |\ \{S=K\})\quad \text{if and only if}\quad  (\bx',K-\bone_{d'}\T\bx')  \in  L_{tf_S(K)}(\bX)
\end{align}
since $f_{\bX'|\{S=K\}}(\bx')=f_{\bX}(\bx',K-\bone_{d'}\T\bx')/f_{S}(K)$ and thus
\begin{align*}
L_t(\bX'\ |\ \{S=K\})&=\{\bx' \in \IR^{d'}: f_{\bX'|\{S=K\}}(\bx')\geq t\}
= \{\bx' \in \IR^{d'}: f_{\bX}(\bx',K-\bone_{d'}\T\bx')\geq t  f_S(K) \}.
\end{align*}
Suppose $\bx',\  \by' \in L_t(\bX'\ |\ \{S=K\})$.
By~\eqref{eq:evuivalence:level:sets}, we have that $(\bx',K-\bone_{d'}\T\bx'),\  (\by',K-\bone_{d'}\T\by') \in  L_{tf_S(K)}(\bX)$.
Since $\bX$ is convex unimodal, $L_{tf_S(K)}(\bX)$ is a convex set.
Therefore, we have, for $\theta \in (0,1)$, that
\begin{align*}
\theta (\bx',K-\bone_{d'}\bx')+(1-\theta)(\by',K-\bone_{d'}\T\by')&=(\theta\bx'+(1-\theta)\by',\ \theta(K-\bone_{d'}\T\bx') + (1-\theta)(K-\bone_{d'}\T\by')\\
&= (\theta\bx'+(1-\theta)\by',\ K-\bone_{d'}\T(\theta\bx' + (1-\theta)\by'))
 \in L_{tf_S(K)}(\bX),
\end{align*}
which implies that $\theta \bx' + (1-\theta)\by' \in L_t(\bX'\ |\ S=K)$ by~\eqref{eq:evuivalence:level:sets}.
\end{enumerate}
 \end{proof}

\subsection*{Proof of Proposition~\ref{prop:tail:density:conditional:distribution:constant:sum}}

\begin{proof}
Let $\tilde \bX = (\bX',K-X_d)$.
Since the density of $\tilde \bX$ is written as $f_{\tilde \bX}(\sqc{x}{d})=f_{\bX}(\sqc{x}{d'},K-x_d)$, we have, by~\eqref{eq:conditional:density:X:constant:sum}, that
\begin{align*}
f_{\bX'|\{S=K\}}(\bx')=\frac{f_{\bX}(\bx',K-\bone_{d'}\T\bx')}{f_S(K)}
=\frac{f_{\tilde \bX}(\bx',\bone_{d'}\T\bx')}{f_S(K)},\quad \bx' \in \IR_{+}^{d'}.
\end{align*}
Therefore, if $\tilde \bX$ has a limit function $\tilde \lambda$, then the density of $\bX'\ |\ \{S=K\}$ satisfies
\begin{align*}
\lim_{t\rightarrow \infty} \frac{f_{\bX'|\{S=K\}}(t\by')}{f_{\bX'|\{S=K\}}(t\bx')}&=
\lim_{t\rightarrow \infty} \frac{f_{\tilde \bX}(t\by',t\bone_{d'}\T\by')}{f_{\tilde \bX}(t\bx',t\bone_{d'}\T\bx')}= \tilde \lambda((\bx',\bone_{d'}\T\bx'),(\by',\bone_{d'}\T\by'))=\lambda'(\bx',\by'),
\end{align*}
for any $\bx',\ \by' \in \IR_{+}^{d'}$ since $(\bx',\bone_{d'}\T\bx'),\ (\by',\bone_{d'}\T\by') \in \IR_{+}^{d}$.
Similarly, if $\tilde \bX$ is MRV($\infty$), then
\begin{align*}
\lim_{t\rightarrow \infty} \frac{f_{\bX'|\{S=K\}}(st\bx')}{f_{\bX'|\{S=K\}}(t\bx')}&=
\lim_{t\rightarrow \infty} \frac{f_{\tilde \bX}(st\bx',st\bone_{d'}\T\bx')}{f_{\tilde \bX}(t\bx',t\bone_{d'}\T\bx')}=\begin{cases}
0, & s>1,\\
\infty, & 0< s < 1,\\
\end{cases}
\end{align*}
for any $s>0$ and $\bx' \in \IR_{+}^{d'}$.
\end{proof}

\subsection*{Proof of Proposition~\ref{prop:mrv:conditional:distribution:elliptical}}

\begin{proof}
Proposition~\ref{prop:ellipticality:conditional:distribution:constant:sum} yields that $\bX'\ |\ \{S=K\}$ follows a $d'$-dimensional elliptical distribution with location vector $\bmu_K$, dispersion matrix $\Sigma_K$ and density generator $g_K$.
If $g$ is regularly varying, then
\begin{align*}
\lim_{t\rightarrow \infty} \frac{f_{\bX'|\{S=K\}}(t\by')}{f_{\bX'|\{S=K\}}(t\bx')}
&= \lim_{t\rightarrow \infty} \frac{g_K\left(\frac{1}{2}(t\by'-\bmu_K)\T\Sigma_K^{-1}(t\by'-\bmu_K)\right)}{g_K\left(\frac{1}{2}(t\bx'-\bmu_K)\T\Sigma_K^{-1}(t\bx'-\bmu_K)\right)}\\
&= \lim_{t\rightarrow \infty} \frac{g\left(\frac{1}{2}t^2(\by'-\bmu_{K}/t)\T\Sigma_K^{-1}(\by'-\bmu_{K}/t)+\Delta_K\right)}{g\left(\frac{1}{2}t^2(\bx'-\bmu_{K}/t)\T\Sigma_K^{-1}(\bx'-\bmu_{K}/t)+\Delta_K\right)}\\
& = \lim_{t\rightarrow \infty} \frac{g(\frac{1}{2}t^2{\by'}\T\Sigma_K^{-1}\by')}{g(\frac{1}{2}t^2{\bx'}\T\Sigma_K^{-1}\bx')}=\lambda_g({\bx'}\T\Sigma_K^{-1}\bx',\ {\by'}\T\Sigma_K^{-1}\by')=\lambda_K(\bx',\by'),
\end{align*}
for any $\bx',\ \by' \in \IR^{d'}$, where the third equality comes from continuity of $g$ and the fourth equality holds since ${\bx'}\T\Sigma_K^{-1}\bx',\ {\by'}\T\Sigma_K^{-1}\by'>0$.
Therefore, $\bX'\ |\ \{S=K\}$ is MRV($\lambda_K$).
For the rapidly varying case,
\begin{align*}
\lim_{t\rightarrow \infty} \frac{f_{\bX'|\{S=K\}}(st\bx')}{f_{\bX'|\{S=K\}}(t\bx')}
 = \lim_{t\rightarrow \infty} \frac{g(\frac{1}{2}t^2s^2{\bx'}\T\Sigma_K^{-1}\bx')}{g(\frac{1}{2}t^2{\bx'}\T\Sigma_K^{-1}\bx')}
=\begin{cases}
0, & s>1,\\
\infty, & 0< s < 1,\\
\end{cases}
\end{align*}
for any $s>0$ and $\bx',\ \by' \in \IR^{d'}$ since $s>1$ if and only if $s^2>1$ and $0<s<1$ if and only if $0<s^2<1$ for $s>0$.
Therefore, $\bX'\ |\ \{S=K\}$ is rapidly varying.
\end{proof}

\subsection*{Proof of Proposition~\ref{prop:properties:mla}}

\begin{proof}
\begin{enumerate}
  \item \emph{Translation invariance}:
  Let $\tilde \bX = \bX + \bc$, $\tilde S = S + \bone_{d}\T\bc$ and $\tilde K=  K + \bone_{d}\T\bc$.
 Since $f_{\bX+\bc}(\bx)=f_{\bX}(\bx-\bc)$, we have that
  \begin{align*}
  f_{\tilde \bX'|\{\tilde S = \tilde K \}}(\tilde \bx')& =\frac{f_{(\tilde \bX,\tilde S)}(\tilde \bx',\tilde K)}{f_{\tilde S}(\tilde K)}
  = \frac{f_{(\bX',S)}(\tilde \bx'-\bc',K)}{f_{S}(K)}
  =  f_{\bX'|\{S=K\}}(\tilde \bx'-\bc').
  \end{align*}
  Therefore, uniqueness of the maximizer of $f_{\bX'|\{S=K\}}$ implies that of $f_{\tilde \bX'|\{\tilde S = \tilde K \}}$, and these maximizers are related via
  $\bK_{\operatorname{M}}[\bX+\bc;\ \mathcal K_d(K+\bone_{d}\T\bc)]=\bK_{\operatorname{M}}[\bX;\mathcal K_d(K)]+\bc$.
  \item \emph{Positive homogeneity}:
  Let $\tilde \bX = c\bX$, $\tilde S = cS $ and $\tilde K=  cK$.
 Since $f_{c\bX}(\bx)=f_{\bX}(\bx/c)$, we have that
  \begin{align*}
  f_{\tilde \bX'|\{\tilde S = \tilde K \}}(\tilde \bx')& =\frac{f_{(\tilde \bX,\tilde S)}(\tilde \bx',\tilde K)}{f_{\tilde S}(\tilde K)}
  = \frac{f_{(\bX',S)}(\tilde \bx'/c,K)}{f_{S}(K)}
  = f_{\bX'|\{S=K\}}(\tilde \bx'/c).
  \end{align*}
  As seen in the case of translation invariance,
  this equality implies that $\tilde \bX \in \mathcal U_d(\tilde K)$ and $\bK_{\operatorname{M}}[\bX;\mathcal K_d(cK)]=c\bK_{\operatorname{M}}[\bX;\mathcal K_d(K)]$.
    \item \emph{Symmetry}: Without loss of generality, consider $i=1$ and $j=2$.
    Let $\tilde \bX = (X_2,X_1,\bX_{-(1,2)})$ and $\tilde S= \bone_d\T \tilde \bX$, where $\bx_{-(1,2)}$ is a shorthand for $(x_3,\dots,x_d)$ for $\bx \in \IR^d$.
    Then $f_{\tilde \bX}(\bx)=f_{\bX}(\tilde \bx)$ for $\bx = ( x_1,x_2,\bx_{-(1,2)}) \in \IR^d$ and $\tilde \bx = (x_2,x_1,\bx_{-(1,2)}) \in \IR^d$.
    Moreover, when $\smash{\bX \deq \tilde \bX}$, we have that $\tilde \bX \in \mathcal U_d(K)$ and $f_{\bX} =f_{\tilde \bX}$. 
    Consequently, we have that
    \begin{align}\label{eq:density:equality:under:symmetry}
 f_{\bX'|\{S=K\}}(\bx')&=\frac{f_{\bX}(\bx',K-\bone_{d'}\T\bx')}{f_S(K)}
    =\frac{f_{\tilde \bX}(\bx',K-\bone_{d'}\T\bx')}{f_{S}(K)}
    =\frac{f_{\bX}(\tilde \bx',K-\bone_{d'}\T\tilde \bx')}{f_{S}(K)}
    =f_{\bX'|\{S=K\}}(\tilde \bx'),
    \end{align}
    where the third equation holds since $\bone_{d'}\T\bx'=\bone_{d'}\T\tilde \bx'$.
Now suppose that $\bK_{\operatorname{M}}[\bX;\mathcal K_d(K)]_1\neq \bK_{\operatorname{M}}[\bX;\mathcal K_d(K)]_2$.
Then two different vectors $\bK_{\operatorname{M}}[\bX;\mathcal K_d(K)]$ and
$(\bK_{\operatorname{M}}[\bX;\mathcal K_d(K)]_2,\ \bK_{\operatorname{M}}[\bX;\mathcal K_d(K)]_{1},\bK_{\operatorname{M}}[\bX;\mathcal K_d(K)]_{-(1,2)})$
attain the maximum of $f_{\bX'|\{S=K\}}$ by ~\eqref{eq:density:equality:under:symmetry}.
 Since $\bK_{\operatorname{M}}[\bX;\mathcal K_d(K)]$ is obtained by the unique maximizer of $f_{\bX'|\{S=K\}}(\bx')$, this leads to a contradiction.
     \item  \emph{Continuity}:
     When $f_n$ is uniformly continuous and bounded for $n=1,2,\dots$, the sequence $(f_n)$ is asymptotically uniformly equicontinuous and bounded in the sense introduced in \citet{sweeting1986converse}.
     Together with the assumption that $\bX_n \rightarrow \bX$ weakly, Theorem 2 of \citet{sweeting1986converse} implies that $f_n \rightarrow f$ pointwise and uniformly in $\IR^d$ for the uniformly continuous density $f$ of $\bX$.
    Define $g_n(\bx')=f_{n}(\bx',K-\bone_{d'}\T\bx')$ for $n=1,2,\dots$ and $g(\bx')=f(\bx',K-\bone_{d'}\T\bx')$, $\bx'\in \IR^{d'}$.
    By~\eqref{eq:conditional:density:X:constant:sum} and since $\bX_n,\ \bX \in \mathcal U_d(K)$,
 the maximizers of $g_n$ and $g$ are uniquely determined.
    Denote them as $\bx_n^\ast =\smash{\argmax_{\bx \in \IR^{d'}} g_n(\bx)}$ and $\bx^\ast =\smash{\argmax_{\bx \in \IR^{d'}} g(\bx)}$.
    By definition of $\bx_n^\ast$, we have that
    \begin{align*}
    g_n(\bx_n^\ast) \geq g_n(\bx) \quad \text{for any}\quad \bx \in \IR^{d'}.
    \end{align*}

    Since $g_n$ converges uniformly to $g$, it holds that
    \begin{align*}
    g(\smash{\limsup_{n\rightarrow \infty} \bx_n^\ast}) \geq g(\bx)\quad \text{and}\quad
    g(\smash{\liminf_{n\rightarrow \infty} \bx_n^\ast}) \geq g(\bx) \quad \text{for any}\quad \bx \in \IR^{d'}.
    \end{align*}
    If $\limsup_{n\rightarrow \infty} \bx_n^\ast>\liminf_{n\rightarrow \infty} \bx_n^\ast$, then two points attain the maximum of $g$, which contradicts the uniqueness of the maximizer of $g$.
    As a consequence, $\limsup_{n\rightarrow \infty} \bx_n^\ast=\liminf_{n\rightarrow \infty} \bx_n^\ast = \lim_{n \rightarrow \infty}\bx_n^\ast = \bx^\ast$ and thus $\lim_{n \rightarrow \infty}\bK_{\operatorname{M}}[\bX_n;\mathcal K_d(K)]=  \bK_{\operatorname{M}}[\bX;\mathcal K_d(K)]$.
\end{enumerate}
\end{proof}

\section{Modality and $s$-concave densities}\label{app:s:concave:densities}

As we saw in Section~\ref{subsec:unimodality:conditional:distribution:constant:sum}, neither joint unimodality nor marginal unimodality implies the other .
However, unimodality is preserved under marginalization for some specific class of densities, so-called $s$-concave densities.
In this appendix we briefly introduce the connection between unimodality and $s$-concavity of the conditional distribution given a constant sum.

\begin{definition}[$s$-concavity]\label{def:s:concavity}
For $s \in \IR$, a density $f$ on $\IR^d$ is called $s$-\emph{concave} on a convex set $A \subseteq \IR^d$ if
\begin{align*}
f(\theta\bx + (1-\theta)\by)\geq M_s(f(\bx),f(\by);\theta),\quad \bx,\by \in A,\quad \theta \in (0,1),
\end{align*}
where $M_s$ is called the \emph{generalized mean} defined, by continuity, as
\begin{align*}
M_s(a,b;\theta)=\begin{cases}
\{\theta a^s + (1-\theta)b^s\}^{1/s}, & 0<s<\infty\text{ or } (-\infty < s<0\ \text{and}\ ab\neq 0),\\
0, & -\infty <s <0 \ \text{and}\ ab = 0,\\
a^{\theta}b^{1-\theta}, & s=0,\\
a \wedge b, & s = - \infty,\\
a \vee b, & s = +\infty,\\
\end{cases}
\end{align*}
for $s\in \IR$, $a,\ b \geq 0$ and $\theta \in (0,1)$.
\end{definition}

Definition~\ref{def:s:concavity} of $s$-concavity is based on densities and can be extended to a measure-based definition for distributions that do not admit densities; see \citet{dharmadhikari1988unimodality}.
For $s=-\infty$, $s$-concavity is also known as \emph{quasi-concavity} and $0$-concavity is also known as \emph{log-concavity}.
By definition, for $0<s<\infty$, $f$ is $s$-concave if and only if $f^s$ is a concave function.
As shown in \citet{dharmadhikari1988unimodality}, the function $s \mapsto M_s(a,b;\theta)$ is increasing for fixed $(a,b;\theta)$.
From this we have that $t$-concavity of $f$ implies $s$-concavity for $s<t$.
Examples of $s$-concave densities include the skew-normal distributions \citep{balkema2010asymptotic}, Wishart distributions, Dirichlet distributions with certain range of parameters \citep{dharmadhikari1988unimodality} and uniform distributions on a convex set in $\IR^d$ \citep{norkin1991alpha}.

Convex unimodality (Definition~\ref{def:concepts:unimodality}) is related to $s$-concavity since a density $f$ is convex unimodal if and only if it is $-\infty$-concave \citep{dharmadhikari1988unimodality}.
Therefore, $f$ is convex unimodal if it is $s$-concave for some $s \in \IR$.
Furthermore, it is straightforward to show that $\bX'\ |\ \{S=K\}$ has an $s$-concave density if $\bX$ has.
As shown in \citet{dharmadhikari1988unimodality} and \citet{saumard2014log}, $s$-concavity is preserved under marginalization, convolution and weak-limit for certain ranges of $s \in \IR$.
Therefore, convex unimodality can also be preserved under these operations if $\bX$ has the $s$-concave the density $f_{\bX}$.

\section{\rev{Dependence of $\bX\ |\ \{S=K\}$ and stochastic orders}}\label{app:dependence:stochastic:order}

\rev{In this appendix, we investigate the dependence, especially the total positivity and its related orders of $\bX'\ |\ \{S=K\}$ implied by those of $\bX$.
To this end, we define the following concepts.
}

\begin{definition}[Multivariate total positivity of order $2$]\label{def:mtp2:mrr2:tp2:order}
Suppose random vectors $\bX$ and $\bY$ have densities $f_{\bX}$ and $f_{\bY}$, respectively.
\begin{enumerate}
\item $\bX$ is said to be \emph{multivariate totally positively ordered of order 2 (MTP2)} if
\begin{align*}
f_{\bX}(\bx)f_{\bX}(\by) \leq f_{\bX}(\bx \wedge \by)f_{\bX}(\bx \vee \by),\quad \text{for all }\bx,\ \by \in \IR^d.
\end{align*}
\item $\bX$ is said to be \emph{multivariate reverse rule of order 2 (MRR2)} if
\begin{align*}
f_{\bX}(\bx)f_{\bX}(\by) \geq f_{\bX}(\bx \wedge \by)f_{\bX}(\bx \vee \by),\quad \text{for all }\bx,\ \by \in \IR^d.
\end{align*}
\item $\bY$ is said to be larger than $\bX$ in \emph{$TP2$-order}, denoted as $\bX\leq_{tp}\bY$ if
\begin{align*}
f_{\bX}(\bx)f_{\bY}(\by) \leq f_{\bX}(\bx \wedge \by)f_{\bY}(\bx \vee \by),\quad \text{for all }\bx,\ \by \in \IR^d.
\end{align*}
\end{enumerate}
\end{definition}

For examples and implied dependence properties of MTP2, MRR2 and TP2 ordered distributions, see \citet{karlin1980classes1} and \citet{karlin1980classes2}.
The following proposition states that the MTP2, MRR2 and TP2 order of $\bX'\ |\ \{\bone_d\T \bX=K\}$ and $\bY'\ |\ \{\bone_d\T \bY=K\}$ are inherited from those of $(\bX',\bone_d\T \bX)$ and $(\bY',\bone_d\T \bY)$.

\begin{proposition}[MTP2, MRR2 and TP2 order of $\bX'\ |\ \{S=K\}$]
\label{prop:mtp2:mrr2:tp2:order:conditional:dist:constant:sum}
Suppose $(\bX',S)$ and $(\bY',T)$ with $S=\bone_d\T\bX$ and $T=\bone_d\T\bY$ have densities $f_{(\bX',S)}$ and $f_{(\bY',T)}$, respectively.
\begin{enumerate}
\item If $(\bX',S)$ is MTP2 (MRR2) then $\bX'\ |\ \{S=K\}$ is MTP2 (MRR2).
\item If $(\bX',S)\leq_{tp} (\bY',T)$ then $\bX'\ |\ \{S=K\}\leq_{tp} \bY'\ |\ \{T=K\}$.
\end{enumerate}
\end{proposition}

\begin{proof}
By~\eqref{eq:conditional:density:X:constant:sum} we have, for $\bx',\ \by' \in \IR^{d'}$, that
\begin{align*}
f_{\bX'|\{S=K\}}(\bx')f_{\bX'|\{S=K\}}(\by')&=\frac{f_{(\bX',S)}(\bx',K)f_{(\bX',S\ )}(\by',K)}{f_S^2(K)}\\
&\leq \frac{f_{(\bX',S)}(\bx' \wedge \by',K\wedge K )f_{(\bX',S)}(\bx'\vee \by',K \vee K)}{f_S^2(K)}\\
&=f_{\bX'|\{S=K\}}(\bx' \wedge \by')f_{\bX'|\{S=K\}}(\bx'\vee \by'),
\end{align*}
which proves the first part on MTP2.
The MRR2 and TP2 parts are shown in similar manners.
\end{proof}

The properties of MTP2 (MRR2) and TP2 order have various implications.
For example, when $\bX'\ |\ \{S=K\}$ is MTP2, then $\bX'\ |\ \{S=K\}$ is \emph{positively associated} in the sense that $\Cov[g(X_i),h(X_j)\ |\ \{S=K\}]\geq 0$ for all increasing functions $g:\IR \rightarrow \IR$ and $h:\IR \rightarrow \IR$.
If $\bX'\ |\ \{S=K\}\leq_{tp} \bY'\ |\ \{T=K\}$, then $\bX'\ |\ \{S=K\} \leq_{st}\bY\ |\ \{T=K\}$, that is, $\E[h(\bX')\ |\ \{S=K\}]\leq \E[h(\bY')\ |\ \{T=K\}]$ for all bounded and increasing functions $h:\IR^{d'}\rightarrow \IR$.
See \citet{muller2002comparison} for more implications of the MTP2, MRR2 and TP2 order.


\section{\rev{Fallacies in risk allocations}}\label{app:properties:not:hold:MLA}

\rev{In this appendix we introduce two properties which intuitively hold but in general do not for the Euler and maximum likelihood allocations.
For a $d$-dimensional random vector $\bX$ and a real number $K \in \IR$, an allocation principle $\bK$ maps $(\bX,K)$ to $\bK(\bX;K) \in \IR^d$ such that $\bone_{d}\T \bK(\bX;K)=K$.}

\begin{enumerate}
\item \emph{Invariance under independence}:\vspace{2mm}\\
For two integers $d,\tilde d\geq 2$, consider a $d$-dimensional random vector $\bX$ with $S=\bone_{d}\T\bX$ and a $\tilde d$-dimensional random vector $\tilde \bX$ with $\tilde S=\bone_{\tilde d}\T\tilde \bX$.
For $K,\tilde K >0$, we call a risk allocation $\bK$ \emph{invariant under independence} if
\begin{align*}
\bK((\bX,\tilde \bX); K+\tilde K) =
(\bK(\bX; K),
\bK(\tilde \bX; \tilde K))
\end{align*}
provided that $\bX$ and $\tilde \bX$ are independent of each other.
This property means that risk allocation problems of multiple portfolios independent of each other can be considered separately.
This property does not hold for MLA since the maximizers of the two functions
\begin{align}\label{eq:counterexample:invariance:independence:LHS}
\nonumber f_{(\bX,\tilde \bX)|\{S+\tilde S=K+\tilde K\}}((\bx,\tilde \bx))&=\frac{f_{(\bX,\tilde \bX)}((\bx,\tilde \bx))\bone_{\{\bone_d\T\bx +\bone_{\tilde d}\T\tilde \bx =K+\tilde K\} }}{f_{S+\tilde S}(K+\tilde K)}
 =\frac{f_{\bX}(\bx)f_{\tilde \bX}(\tilde \bx)\bone_{\{\bone_d\T\bx +\bone_{\tilde d}\T\tilde \bx =K+\tilde K\} }}{f_{S+\tilde S}(K+\tilde K)}\\
 &\propto f_{\bX}(\bx) f_{\tilde \bX}(\tilde \bx)\bone_{\{\bigcup_{\{(k,\tilde k)\in\IR^{2},\ k + \tilde k =K + \tilde K\}}\{\bone_d\T\bx= k\}\cap \{\bone_{\tilde d} \tilde \bx= \tilde k\}\}  }
\end{align}
and
\begin{align}\label{eq:counterexample:invariance:independence:RHS}
f_{\bX|\{S=K\}}(\bx)
 f_{\tilde \bX|\{\tilde S=\tilde K\}}(\tilde \bx) \propto
f_{\bX}(\bx)\bone_{\{\bone_d\T\bx= K\}}f_{\tilde \bX}(\tilde \bx)\bone_{\{\bone_{\tilde d} \tilde \bx= \tilde K\}}
\end{align}
are in general different.
For example, let $d=d'$ and $\bX$ and $\tilde \bX$ be two independent and identically distributed standard normal distributions.
Then the maximum of~\eqref{eq:counterexample:invariance:independence:LHS} is attained at $(K+\tilde K) \bone_{2d} /2d$ whereas that of ~\eqref{eq:counterexample:invariance:independence:RHS} is attained at $(K\bone_d/d,\tilde K\bone_d/d)$.
The two vectors are not equal unless $K=\tilde K$.
\rev{In this example, the Euler allocation provides the same allocated capitals as MLA. 
Therefore, neither Euler allocation nor MLA satisfies invariance under independence.}
\item \emph{Additivity under convolution}:\vspace{2mm}\\
Consider two independent $d$-dimensional random vectors $\bX$ and $\tilde \bX$ with $S=\bone_{d}\T\bX$ and $\tilde S=\bone_{d}\T\tilde \bX$.
For $K,\ \tilde K >0$,
we call an allocation $\bK$ \emph{additive under convolution} if
\begin{align*}
\bK(\bX+\tilde \bX; K+\tilde K) =
\bK(\bX; K)+
\bK(\tilde \bX;\tilde K).
\end{align*}
Neither Euler allocation nor MLA satisfies this property.
For example, let $\bX \sim \N_d(\bmu,\Sigma)$ and $\tilde \bX \sim \N_d(\tilde \bmu,\tilde \Sigma)$ be two independent normal random vectors for $\bmu,\tilde \bmu \in \IR^d$ and $\Sigma, \tilde \Sigma \in \mathcal M_{+}^{d\times d}$.
By Proposition~\ref{prop:ellipticality:conditional:distribution:constant:sum}, Equation~\eqref{eq:location:dispersion:parameter:conditional:dist:constant:sum} and Proposition~\ref{prop:unimodality:conditional:distribution:constant:sum} Part 1, we have that
\begin{align*}
\bK_{\operatorname{M}}(\bX; \mathcal K_{d}(K))= \bmu' + \frac{K - \mu_S}{\sigma_S^2}(\Sigma\bone_{d})'\quad\text{and}\quad
\bK_{\operatorname{M}}(\tilde \bX; \mathcal K_{d}(\tilde K))= \tilde \bmu' + \frac{\tilde K - \mu_{\tilde S}}{\sigma_{\tilde S}^2}(\tilde \Sigma\bone_{d})'.
\end{align*}
Similarly, since $\bX+\tilde \bX \sim \N_d(\bmu+\tilde \bmu,\Sigma+\tilde \Sigma)$, we have that $\sigma_{S+\tilde S}^2=\sigma_S^2 + \sigma_{\tilde S}^2$ and that
\begin{align*}
\bK_{\operatorname{M}}(\bX+\tilde \bX; \mathcal K_{d}(K+\tilde K))&= \bmu' + \tilde \bmu'  + \frac{K+\tilde K - (\mu_{S}+\mu_{\tilde S})}{\sigma_S^2+\sigma_{\tilde S}^2}((\Sigma+\tilde \Sigma)\bone_{d})'\\
&= \bmu' + \tilde \bmu'  + \left(
\frac{\sigma_S^2}{\sigma_S^2+\sigma_{\tilde S}^2} \frac{K -\mu_{S}}{\sigma_S^2} +
\frac{\sigma_{\tilde S}^2}{\sigma_S^2+\sigma_{\tilde S}^2}\frac{\tilde K -\mu_{\tilde S}}{\sigma_{\tilde S}^2}\right)
((\Sigma+\tilde \Sigma)\bone_{d})',
\end{align*}
which is not equal to $
\bK_{\operatorname{M}}(\bX; \mathcal K_{d}(K))+
\bK_{\operatorname{M}}(\tilde \bX; \mathcal K_{d}(\tilde K))$ unless, for instance, $\Sigma=\tilde \Sigma$.
\rev{Since Euler and maximum likelihood allocations coincide under ellipticality, the same statement holds for Euler allocations.}
\end{enumerate}

\section{\rev{Further properties of the multimodality adjustment}}\label{app:sec:further:properties:multimodality:adjustment}

In this section we study further properties of the multimodality adjustment $\bK_{\bw,\mathcal X, \Lambda}$ introduced in Section~\ref{subsec:multimodality:adjustment}.
To clarify the relationship between $\bK_{\bw,\mathcal X, \Lambda}$, the total capital $K$ and the loss distribution of $\bX$, define $\bK_{\bw,\mathcal X, \Lambda}[\bX; \mathcal K_d(K)]$ and $\bar \bK_{\bw,\mathcal X}[\bX; \mathcal K_d(K)]$ to be the multimodality-adjusted allocated capitals~\eqref{eq:multimodality:adjusted:risk:allocation} and their first term $\sum_{m=1}^M w_m \bK_m$, respectively, with $\mathcal X$ being the set of local modes $\bK_1,\dots,\bK_M$ of $\bx \mapsto f_{\bX}(\bx)\bone_{\{\bx \in \mathcal K_d(K)\}}$ (assumed to be a finite set) and with $w_m \propto f_{\bX}(\bK_m)$.
To this end, we adopt the following definition of local modes.

\begin{definition}[Local modes]\label{def:local:modes}
For an $\IR_{+}$-valued function $f$ on $\IR^d$, $\bx \in \IR^d$ is called a \emph{local mode} of $f$ if there exists $\epsilon>0$ such that
\begin{align}\label{eq:local:modes:definition}
f(\bx)\geq f(\by)\quad \text{for all}\quad \by \in \mathcal N_{\epsilon}(\bx),
\end{align}
where $\mathcal N_{\epsilon}(\bx)=\{\bz \in \IR^d: ||\bz-\bx|| < \epsilon\}$.
If~\eqref{eq:local:modes:definition} holds for any $\epsilon>0$, then $\bx$ is called a \emph{global mode} of $f$.
\end{definition}

Properties of $\bK_{\bw,\mathcal X, \Lambda}[\bX; \mathcal K_d(K)]$ are then summarized as follows.

\begin{enumerate}
\item \emph{Translation invariance}:
We show that  $\bK_{\bw, \mathcal X,\Lambda}$ is translation invariant in the sense that
\begin{align*}
\bK_{\bw, \mathcal X,\Lambda}[\bX+\bc;\mathcal K_d(K+\bone_d\T\bc)]=\bK_{\bw, \mathcal X,\Lambda}[\bX;\mathcal K_d(K)] + \bc\quad \text{for}\quad \bc \in \IR^d.
\end{align*}
To show this, notice that local modes of $\bx \mapsto f_{\bX+\bc}(\bx)\bone_{\{\bx \in \mathcal K_d(K+\bone_d\T\bc)\}}$ are given by $\bK_m+\bc$, $m=1,\dots,M$, if $\bK_m$, $m=1,\dots,M$, are the local modes of $\bx \mapsto f_{\bX}(\bx)\bone_{\{\bx \in \mathcal K_d(K)\}}$.
Since $w_m = f_{\bX}(\bK_m) = f_{\bX+\bc}(\bK_m + \bc)$, the probability weight assigned to the $m$th scenario does not change from $(\bX,\mathcal K_d(K))$ to $(\bX+\bc,\mathcal K_d(K+\bone_d\T\bc))$ for all $m=1,\dots,M$.
Therefore, $\bar \bK_{\bw,\mathcal X}[\bX+\bc; \mathcal K_d(K+\bone_d\T\bc)]=\bar \bK_{\bw,\mathcal X}[\bX; \mathcal K_d(K)]+\bc$ and thus
\begin{align*}
\bK_{\bw,\mathcal X,\Lambda}[\bX+\bc;\mathcal K_d(K+\bone_d\T\bc)]& =\bar \bK_{\bw,\mathcal X}[\bX+\bc; \mathcal K_d(K+\bone_d\T\bc)] \\
& \hspace{5mm}+ \sum_{m=1}^M w_m \blambda_m \circ (\bK_m+\bc-\bar \bK_{\bw,\mathcal X}[\bX+\bc; \mathcal K_d(K+\bone_d\T\bc)])^{+}\\
&= \bar \bK_{\bw,\mathcal X}[\bX; \mathcal K_d(K)]+\bc + \sum_{m=1}^M w_m \blambda_m \circ (\bK_m-\bar \bK_{\bw,\mathcal X}[\bX; \mathcal K_d(K)])^{+}\\
&= \bK_{\bw,\mathcal X,\Lambda}[\bX;\mathcal K_d(K)]+\bc,
\end{align*}
which shows translation invariance.
\item \emph{Positive homogeneity}:
Multimodality-adjusted allocated capitals are positive homogeneous in the sense that
\begin{align*}
\bK_{\bw, \mathcal X,\Lambda}[c\bX;\mathcal K_d(cK)] =c\bK_{\bw, \mathcal X,\Lambda}[\bX;\mathcal K_d(K)] \quad \text{for}\quad c >0.
\end{align*}
This can be checked similarly as translation invariance.
The local modes of $\bx \mapsto f_{c\bX}(\bx)\bone_{\{\bx \in \mathcal K_d(cK)\}}$ are given by $c\bK_m$, $m=1,\dots,M$, if $\bK_m$, $m=1,\dots,M$, are the local modes of $\bx \mapsto f_{\bX}(\bx)\bone_{\{\bx \in \mathcal K_d(K)\}}$.
Moreover, the probability weight assigned to the $m$th scenario does not change from $(\bX,\mathcal K_d(K))$ to $(c\bX,\mathcal K_d(cK))$ since $w_m = f_{\bX}(\bK_m) = f_{c\bX}(c \bK_m)$ for all $m=1,\dots,M$.
Therefore, $\bar \bK_{\bw,\mathcal X}[c \bX; \mathcal K_d(c K)]=c \bar \bK_{\bw,\mathcal X}[\bX; \mathcal K_d(K)]$ and thus
\begin{align*}
\bK_{\bw,\mathcal X,\Lambda}[c \bX;\mathcal K_d(c K)]& =\bar \bK_{\bw,\mathcal X}[c \bX; \mathcal K_d(cK)] \\
& \hspace{5mm}+ \sum_{m=1}^M w_m \blambda_m \circ (c \bK_m-\bar \bK_{\bw,\mathcal X}[c\bX; \mathcal K_d(c K)])^{+}\\
&= c \bar \bK_{\bw,\mathcal X}[\bX; \mathcal K_d(K)] + c \sum_{m=1}^M w_m \blambda_m \circ (\bK_m-\bar \bK_{\bw,\mathcal X}[\bX; \mathcal K_d(K)])^{+}\\
&= c \bK_{\bw,\mathcal X,\Lambda}[\bX;\mathcal K_d(K)],
\end{align*}
which shows positive homogeneity.
\item \emph{Riskless asset}:
Multimodality-adjusted allocated capitals satisfy the riskless asset property in the following sense.
Suppose that $X_j=c_j \in \IR$ a.s.\ for $j \in I \subseteq \{1,\dots,d\}$ and that $\bX_{-I}=(X_j,j\in \{1,\dots,d\}\backslash I)$ admits a density $f_{\bX_{-I}}$.
As we discussed in Section~\ref{subsubsec:properties:mla}, any realization $\bx$ of $\bX\ |\ \{S=K\}$ satisfies $\bx_{I}=\bc$ where $\bc=(c_j; j \in I)$, and the likelihood of $\bx$ is quantified through the density $f_{\bX_{-I}|\{\bone_{|-I|}\T\bX_{-I}=K-\bone_{|I|}\T\bc\}}(\bx_{-I})$.
Therefore, reasonable choices of the scenarios $\bK_1,\dots,\bK_M \in \mathcal K_d(K)$ are such that
$(\bK_{m})_{I}=\bc$ and $(\bK_{m})_{-I}$ are local modes of $\bX_{-I}\ |\ \{\bone_{|-I|}\T\bX_{-I}=K-\bone_{|I|}\T\bc\}$.
In this case, we have that $(\bar \bK_{\bw,\mathcal X})_{I}=\bc$ and the multimodality adjustment yields
\begin{align*}
\left(\sum_{m=1}^M w_m \blambda_m \circ (\bK_m-\bar \bK_{\bw, \mathcal X} )^{+}\right)_{I}
=\sum_{m=1}^M w_m (\blambda_m)_{I} \circ (\bc - \bc)^{+}=\bzero_{|I|}.
\end{align*}
Therefore, it holds that $\bK_{\bw,\mathcal X,\Lambda}[\bX; \mathcal K_d(K)]_{I}=\bc$ if $\bX_{I}=\bc$ a.s.
\item \emph{Symmetry}:
For a reasonable choice of $\Lambda$, the multimodality-adjusted allocated capitals satisfy the symmetry property, that is, $\bK_{\bw, \mathcal X,\Lambda}[\bX;\mathcal K_d(K)]_i=\bK_{\bw, \mathcal X,\Lambda}[\bX;\mathcal K_d(K)]_j$ for $i,j \in \{ 1,\dots,d\}$, $i\neq j$, if $\smash{\bX\deq \tilde \bX}$ where $\tilde \bX$ is a $d$-dimensional random vector satisfying $\tilde X_j=X_i$, $\tilde X_i=X_j$ and $\tilde X_k=X_k$ for all $k \in \{1,\dots,d\}\backslash\{i,j\}$.
For any $\bx \in \IR^d$, denote by $\tilde \bx$ a $d$-dimensional vector such that $\tilde x_j=x_i$, $\tilde x_i=x_j$ and $\tilde x_k=x_k$ for all $k \in \{1,\dots,d\}\backslash\{i,j\}$.
To show the symmetry, suppose that $\bK \in \mathcal K_d(K)$ is a local mode of $\bx \mapsto f_{\bX}(\bx)\bone_{\{\bx \in \mathcal K_d(K)\}}$.
Then $\tilde \bK$ is also a local mode of $\bx \mapsto f_{\bX}(\bx)\bone_{\{\bx \in \mathcal K_d(K)\}}$ since $f_{\bX}(\by)=f_{\tilde \bX}(\tilde \by)=f_{\bX}(\tilde \by)$ for any $\by \in \IR^d$ by $\bX\deq \tilde \bX$, and thus $\tilde \bK$ satisfies~\eqref{eq:local:modes:definition} for some $\epsilon>0$.
Therefore, any element $\bK_m=(K_{m,1},\dots,K_{m,d})$ in $\mathcal X$ satisfies either (1) $K_{m,i}=K_{m,j}$ or (2) there exists a unique element $\bK_{m'} \in \mathcal X$ such that $K_{m,i}\neq K_{m,j}$, $K_{m',i}\neq K_{m',j}$, $K_{m',i}=K_{m,j}$, $K_{m,i}=K_{m',j}$ and $K_{m,k}=K_{m',k}$ for all $k \in \{1,\dots,d\}\backslash\{i,j\}$.
For such a pair of indices $(m,m')$, it holds that $w_m=w_{m'}$ since $w_m \propto f_{\bX}(\bK_m)$ and $w_{m'} \propto f_{\bX}(\bK_{m'})=f_{\tilde \bX}(\bK_{m})=f_{\bX}(\bK_{m})$ by $\bX\deq \tilde \bX$.
Therefore, if $\blambda_m=\blambda_{m'}$ holds for all pairs of $(m,m')$, then
the symmetry $\bK_{\bw, \mathcal X,\Lambda}[\bX;\mathcal K_d(K)]_i=\bK_{\bw, \mathcal X,\Lambda}[\bX;\mathcal K_d(K)]_j$ holds.

\end{enumerate}

The continuity property $\lim_{n\rightarrow \infty}\bK_{\bw, \mathcal X,\Lambda}[\bX_n;\mathcal K_d(K)]=\bK_{\bw, \mathcal X,\Lambda}[\bX;\mathcal K_d(K)]$ for a given $\bX_n$ and $\bX$ such that $\bX_n$ converges to $\bX$ weakly
may not be straightforward to verify since a limit of multimodal distributions can be unimodal, and more generally, the number of scenarios may change in $n$.


\section{Simulation of $\bX\ |\ \{S=K\}$ with MCMC methods}\label{app:simulation:with:MCMC}

Efficient simulation of the conditional distribution of $\bX$ given a constant sum $\{S=K\}$ for $K\in \IR$ is a challenging task in general.
In Section~\ref{subsec:motivating:example} and Section~\ref{sec:numerical:experiments}, the constraint $\{S=K\}$ was replaced by $\{K-\delta <S<K+\delta\}$ for a small $\delta>0$ so that $\P(K-\delta <S<K+\delta)>0$.
However, this modification distorts the conditional distribution $\bX\ |\ \{S=K\}$ and the resulting estimates of risk allocations suffer from inevitable biases.
To overcome this issue,
we briefly review MCMC methods, specifically the \emph{Metropolis-Hastings (MH)} algorithm, and then demonstrate their efficiency for simulating $\bX\ |\ \{S=K\}$.

\subsection{MCMC methods}\label{app:subsec:simulation:HMC}

As mentioned in Section~\ref{subsec:capital:allocation}, it suffices to simulate $\bX'\ |\ \{S=K\}=(X_1,\dots,X_{d'})\ |\ \{S=K\}$ for $d'=d-1$.
Assume that $\bX'\ |\ \{S=K\}$ admits a density~\eqref{eq:conditional:density:X:constant:sum}.
We call this density the \emph{target density} and denote it as $\pi$.
In the MCMC approach, a Markov chain is constructed such that its stationary distribution is $\pi$.
Constructing such a Markov chain can be achieved by the MH algorithm as we now explain.
From the current state $\bX'_n$, a candidate $\bY'_{n}$ of the next state is simulated from $q(\bX'_n,\cdot)$ where $q(\bx',\by')$, $\bx',\by'\in\IR^{d'}$ is called the \emph{proposal density} satisfying the two conditions that (i) $\bx' \mapsto q(\bx',\by')$ is measurable for all $\by' \in \IR^{d'}$, and (ii) $\by' \mapsto q(\bx',\by')$ is a density function for all $\bx' \in \IR^{d'}$.
The candidate is accepted, that is, $\bX'_{n+1}=\bY'_n$, with probability $\alpha(\bX'_n,\bY'_{n})$ where
\begin{align}\label{eq:acceptance:probability}
\alpha(\bx',\by')=1 \wedge \frac{q(\bx',\by')\pi(\by')}{q(\by',\bx')\pi(\bx')}=1 \wedge \frac{q(\bx',\by')f_{\bX}(\by',K-\bone_{d'}\T\by')}{q(\by',\bx')f_{\bX}(\bx',K-\bone_{d'}\T\bx')},
\end{align}
and otherwise the chain stays at the current state $\bX'_{n+1}=\bX'_n$.
Calculation of the \emph{acceptance probability}~\eqref{eq:acceptance:probability} is often possible since it does not depend on $f_{S}(K)$.
The resulting Markov chain is shown to have $\pi$ as a stationary distribution and thus $\bX_1',\bX_2',\dots$ can be used as samples from $\pi$ in order to estimate Euler and maximum likelihood allocations.

An appropriate choice of $q$ is important since MCMC samples are typically positively correlated due to the acceptance-rejection procedure.
To reduce positive correlation among MCMC samples, $q$ must reflect properties of $\pi$, such as the shape of its support, modality and tail behavior, for maintaining high acceptance probability $\alpha$.
First, the support of $\pi$ must be taken into account since a candidate outside of the support of $\pi$ is immediately rejected.
Second, heavy-tailed target distribution $\pi$ requires specific design of $q$ since most standard MCMC methods such as random walk MH, independent MH, Gibbs samplers and the Hamiltonian Monte Carlo methods cannot guarantee the theoretical convergence when $\pi$ is heavy-tailed.
Finally, multimodality of $\pi$ also requires to be handled specifically since the chain needs to traverse from one mode to another to sample from the entire support of $\pi$.

\subsection{An application of MCMC to core allocation}\label{app:subsec:application:core:allocation}

In this section we compute the Euler allocation and MLA on the restricted set of allocations called the \emph{(atomic) core} defined by
\begin{align*}
\mathcal K_d^\text{C}(K; r)=\{\bx \in \IR^d: \bone_d\T\bx=K,\ \blambda\T\bx \leq r(\blambda),\ \blambda \in \{0,1\}^d\}\subseteq \mathcal K_d(K),
\end{align*}
where $K$ is a given total capital and $r:\{0,1\}^d \rightarrow \IR$ is called a \emph{participation profile function} typically determined as $r(\blambda)=\varrho(\blambda\T\bX)$ for a $d$-dimensional loss random vector $\bX$ and a risk measure $\varrho$.
We call an element of $\mathcal K_d^\text{C}(K; r)$ a \emph{core allocation}.
As explained in \citet{denault2001coherent}, core allocations possess an important property as risk allocations, that is, any subportfolio of $\bX=(\sqc{X}{d})$ of the form $(\lambda_1 X_1,\dots,\lambda_d X_d)$ gains benefit of capital reduction by managing risk as a portfolio $\bX$.
In fact, for a participation profile $\blambda=(\lambda_1,\dots,\lambda_d)$ where $\lambda_j\in \{0,1\}$ represents the presence $(\lambda_j=1)$ or the absence $(\lambda_j=0)$ of the $j$th entity, the total amount of capital required to cover the loss $\blambda\T\bX$ is $\blambda\T\bx$ for an allocation $\bx \in \mathcal K_d(K)$.
The value $r(\lambda)=\varrho(\blambda\T\bX)$ is interpreted as a stand-alone capital that would have been required if the total loss $\blambda\T\bX$ had been managed individually.
 Therefore, under the core allocation $\bx \in \mathcal K_d^\text{C}(K; r)$, the
subportfolio $(\lambda_1 X_1,\dots,\lambda_d X_d)$ gains benefit of capital reduction by $\blambda\T\bx$ in comparison to $r(\blambda)$.

Given $K$, $r$ and the joint loss $\bX$, we are interested in calculating the core-compatible versions of Euler allocation $\E[\bX\ |\  \{\bX \in \mathcal K_d^\text{C}(K; r)\}]$, MLA $\bK_{\operatorname{M}}[\bX;\mathcal K_d^\text{C}(K; r)]$ and local modes of $f_{\bX| \{\bX \in \mathcal K_d^\text{C}(K; r)\}}$ if they exist.
However, generating a large number of samples from $\bX'\ |\ \{\bX \in \mathcal K_d^\text{C}(K; r)\}$ is computationally involved since an unconditional sample $\bX$ is first filtered by the condition $\bX \in \{\bx \in \IR^d: K-\delta<\bone_d\T\bx<K+\delta\}=\mathcal K_d(K,\delta)$ for a small $\delta>0$, and then filtered again by the core condition $\blambda\T\bX \leq r(\blambda)$ for all possible $\blambda\in\{0,1\}^d$.
To overcome this issue, we utilize the \emph{Hamiltonian Monte Carlo (HMC) method with reflection} to directly simulate $\smash{f_{\bX'| \{\bX \in \mathcal K_d^\text{C}(K; r)\}}}$.
Note that the support of $\smash{\bX'\ |\  \{\bX \in \mathcal K_d^\text{C}(K; r)\}}$ is a projection of $\mathcal K_d^\text{C}(K; r)$ onto $\IR^{d'}$, which is an intersection of finite number of hyperplanes $\{\bx' \in \IR^{d'}: \blambda\T(\bx',K-\bone_{d'}\T\bx') \leq r(\blambda)\}$ for $\blambda \in \{0,1\}^d$.
In the HMC method, a candidate is proposed according to the so-called Hamiltonian dynamics, and the chain reflects at the boundaries $\{\bx' \in \IR^{d'}: \blambda\T(\bx',K-\bone_{d'}\T\bx')=r(\blambda)\}$, $\blambda \in \{0,1\}^d$, so that it does not violate the support constraint; see \citet{koike2020markov} for details.

For a numerical experiment, let $\bX\sim t_{\nu}(\bzero_d,P)$ with $d=3$, $\nu=5$ and $P=(\rho_{ij})$ being a correlation matrix with $\rho_{12}=\rho_{23}=1/3$ and $\rho_{13}=2/3$.
For $p=0.99$, we set $r(\blambda)=\VaR{p}{\blambda\T\bX}$ for $\blambda \in \{0,1\}^3$ and $K=r(\bone_3)$.
For $\delta=0.001$, we first generate $N_{\text{MC}}=10^6$ samples from $\bX$ and estimate $K$ and $(r(\blambda),\blambda\in\{0,1\}^3)$ from these samples.
Then we extract samples of $\bX$ falling in the region
\begin{align*}
\mathcal K_d^\text{C}(K, \delta ;r)=\mathcal K_d(K,\delta)\cap\{ \bx \in \IR^d: \blambda\T\bx \leq r(\blambda),\ \blambda \in \{0,1\}^3\backslash \{\bone_3\}\}.
\end{align*}
Figure~\ref{fig:plots:core:mcmc} (a) shows the first two components of the MC samples from $\bX$ and the conditional samples falling in $\mathcal K_d^\text{C}(K, \delta ;r)$.
Among the $N_{\text{MC}}=10^6$ samples, $2000$ samples were contained in $\mathcal K_d(K,\delta)$ and only $189$ samples fell in $\mathcal K_d^\text{C}(K, \delta ;r)$.
Therefore, this crude simulation method is not efficient since 99.98\% of the unconditional samples are discarded.

Instead, we conduct an MCMC simulation to generate $N_{\text{MCMC}}=10^4$  samples directly from $\bX\ |\ \{\bX \in \mathcal K_d^\text{C}(K;r)\}$.
Hyperparameters of the HMC method are estimated based on the 189 MC samples; see \citet{koike2020markov}for details.
The resulting stepsize and integration time are $\eps=0.105$ and $T=24$, respectively.
It took 49.534 seconds to simulate a Markov chain with length $N_{\text{MCMC}}=10^4$ on a MacBook Air with 1.4 GHz Intel Core i5 processor and 4 GB 1600 MHz of DDR3 RAM.
The resulting acceptance rate was 0.866 and serial correlations were below 0.03 at lag $1$.
Based on these inspections, we conclude that the MCMC method performed correctly.
The first 3000 MCMC samples of $\bX'\ |\ \{\bX \in \mathcal K_d^\text{C}(K;r)\}$ are plotted in Figure~\ref{fig:plots:core:mcmc} (b).

\begin{figure}[t]
  \centering
  \vspace{-35mm}
  \includegraphics[width=15 cm]{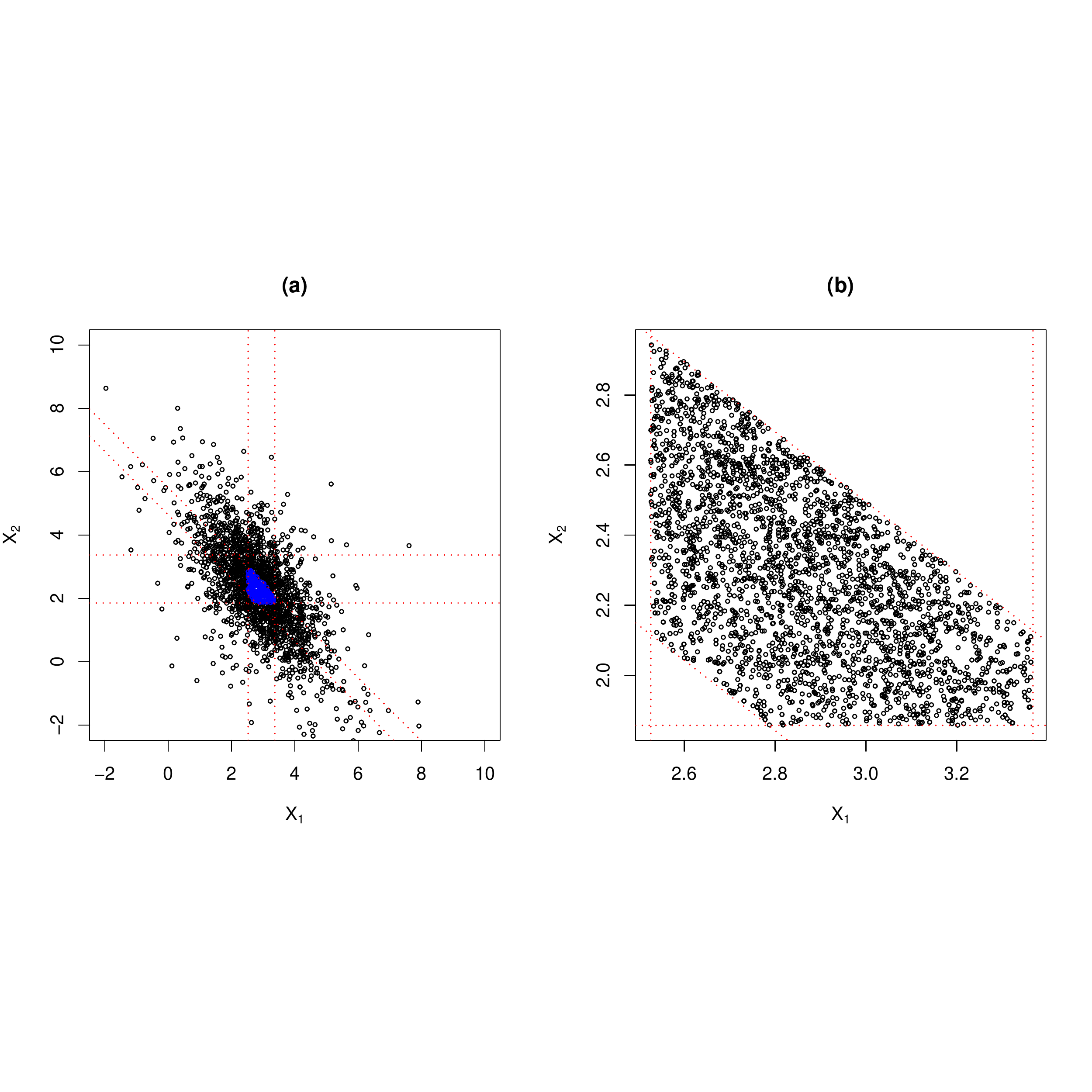}\vspace{-35mm}
  \caption{Scatter plots of (a) MC samples from $\bX'\ | \ \{\bX\in \mathcal K_d(K,\delta)\}$ (black) and $\bX'\ | \ \{\bX\in \mathcal K_d^\text{C}(K,\delta;r)\}$ (blue), and of (b) MCMC samples from $\bX'\ | \ \{\bX\in \mathcal K_d^\text{C}(K;r)\}$ (black) where $\bX \sim t_\nu(\bzero_d,P)$ with $d=3$, $\nu=5$ and $P=(\rho_{i,j})$ being a correlation matrix with $\rho_{1,2}=\rho_{2,3}=1/3$ and $\rho_{1,3}=2/3$,
  $r(\blambda)=\VaR{p}{\blambda\T\bX}$ with $p=0.99$ for $\blambda \in \{0,1\}^3$, $K=r(\bone_3)$ and $\delta=0.001$.
  Red lines indicate $\{\bx' \in \IR^2: \blambda\T(\bx',K-\bone_2\T\bx')=r(\blambda)\}$ for $\blambda \in \{0,1\}^3$.}
  \label{fig:plots:core:mcmc}
\end{figure}

By Proposition~\ref{prop:ellipticality:conditional:distribution:constant:sum}, $\bX'\ |\ \{\bX\in\mathcal K_d(K)\}$ follows a multivariate Student $t$ distribution, and thus the mode of this conditional distribution is uniquely determined by $\bK_{\operatorname{M}}[\bX;\mathcal K_d(K)]=\E[\bX\ |\ \{\bX\in\mathcal K_d(K)\}]$ by Proposition~\ref{prop:unimodality:conditional:distribution:constant:sum} Part 1.
Moreover, when this point is contained in the core $\bK_d^\text{C}(K;r)$, we have that $\bK_{\operatorname{M}}[\bX;\mathcal K_d(K)]=\bK_{\operatorname{M}}[\bX;\mathcal K_d^\text{C}(K;r)]$ since 
$\mathcal K_d^\text{C}(K;r) \subseteq \mathcal K_d(K)$.
We check these observations numerically by calculating the corresponding estimates.

Table~\ref{table:estimators:MC:MCMC} summarizes the MC and MCMC estimates and
standard errors of the Euler and maximum likelihood allocations on
$\mathcal K_d(K)$ and those on the atomic core $\mathcal K_d^\text{C}(K;r)$. 
MC and MCMC estimates are calculated based on the samples in
Figure~\ref{fig:plots:core:mcmc} (a), and 
based on those in Figure~\ref{fig:plots:core:mcmc} (b), respectively.  As expected by theory,
the MC estimates of $\bK_{\operatorname{M}}[\bX;\mathcal K_d(K)]$ and
$\E[\bX\ |\ \{\bX\in\mathcal K_d(K)\}]$ were close to each other.  We can also
observe that the MC and MCMC estimates are close to each other for all the
estimators.  The standard errors of the MCMC estimator of
$\E[\bX\ |\ \{\bX\in\mathcal K_d^\text{C}(K;r)\}]$ were smaller than those of
the MC estimator because of sample efficiency.  Provided that
$\hat{\bK}_{\operatorname{M}}^{\operatorname{MC}}[\bX;\mathcal K_d(K)]$ belongs
to the core $\mathcal K_d^\text{C}(K;r)$, we expect an estimate of
$\bK_{\operatorname{M}}[\bX;\mathcal K_d^\text{C}(K;r)]$ to be close to
$\hat{\bK}_{\operatorname{M}}^{\operatorname{MC}}[\bX;\mathcal K_d(K)]$.
Although this was the case for both of the MC and MCMC estimates of
$\bK_{\operatorname{M}}[\bX;\mathcal K_d^\text{C}(K;r)]$, the MCMC estimate was
slightly closer to
$\hat{\bK}_{\operatorname{M}}^{\operatorname{MC}}[\bX;\mathcal K_d(K)]$ than the
MC estimate.  Consequently, the MCMC estimator of
$\bK_{\operatorname{M}}[\bX;\mathcal K_d^\text{C}(K;r)]$ can be said to be less biased than
the MC estimator.

\begin{table}[t]
\centering
\caption{Monte Carlo (MC) and Markov chain Monte Carlo (MCMC) estimates and
  standard errors of the Euler and maximum likelihood allocations on
  $\mathcal K_d(K)$ and those on the atomic core $\mathcal K_d^\text{C}(K;r)$.
  The MC sample size of $\bX$ is $N_{\operatorname{MC}}=10^6$
  and the MCMC sample size of $\bX\ | \ \{\bX \in \mathcal K_d^\text{C}(K;r)\}$ is
  $N_{\operatorname{MCMC}}=10^4$.}\label{table:estimators:MC:MCMC}\vspace{5mm}
\begin{tabular}{l rrr rrr}
  \toprule
  & \multicolumn{3}{c}{Estimator} & \multicolumn{3}{c}{Standard error}\\
  \cmidrule(lr{0.4em}){2-4}\cmidrule(lr{0.4em}){5-7}
  & \multicolumn{1}{c}{$X_1$} & \multicolumn{1}{c}{$X_2$} & \multicolumn{1}{c}{$X_3$} & \multicolumn{1}{c}{$X_1$} & \multicolumn{1}{c}{$X_2$} & \multicolumn{1}{c}{$X_3$}\\
  \midrule
$\hat{\E}^{\operatorname{MC}}[\bX\ |\ \{\bX\in\mathcal K_d(K)\}]$ & 2.865 & 2.310 & 2.846& 0.026&0.034& 0.026\\
$\hat{\bK}_{\operatorname{M}}^{\operatorname{MC}}[\bX;\mathcal K_d(K)]$ & 2.861&2.366&2.793& \multicolumn{1}{c}{--} & \multicolumn{1}{c}{--} & \multicolumn{1}{c}{--}\\[3mm]
$\hat{\E}^{\operatorname{MC}}[\bX\ |\ \{\bX\in\mathcal K_d^\text{C}(K;r)\}]$&2.852& 2.267 &2.903& 0.016 &0.019 & 0.016\\
$\hat{\bK}_{\operatorname{M}}^{\operatorname{MC}}[\bX;\mathcal K_d^\text{C}(K;r)]$&2.838&2.262&2.920& \multicolumn{1}{c}{--} & \multicolumn{1}{c}{--} & \multicolumn{1}{c}{--}\\[3mm]
$\hat{\E}^{\operatorname{MCMC}}[\bX\ |\ \{\bX\in\mathcal K_d^\text{C}(K;r)\}]$&
2.876 & 2.269 & 2.877&0.002 & 0.003 & 0.002 \\
$\hat{\bK}_{\operatorname{M}}^{\operatorname{MCMC}}[\bX;\mathcal K_d^\text{C}(K;r)]$&2.866&2.283& 2.871& \multicolumn{1}{c}{--} & \multicolumn{1}{c}{--} & \multicolumn{1}{c}{--}\\
  \bottomrule
\end{tabular}
\end{table}

\end{document}